\documentclass[11pt,a4paper]{article}

\usepackage{a4wide}

\usepackage{graphicx}
\usepackage{subcaption}
\usepackage[dvipsnames]{xcolor}
\usepackage{pgfplots}
\usepackage[utf8]{inputenc}
\usepackage[english]{babel}
\usepackage{amsmath}
\usepackage{amssymb}
\usepackage{enumerate}
\usepackage{amsthm}
\usepackage{ulem}
\usepackage{url}
\usepackage{enumerate}
\usepackage{blindtext}
\usepackage{array}
\usepackage{xcolor}
\usepackage{euscript}
\usepackage{comment}
\usepackage{enumitem}
\setlist{leftmargin=5mm}

\usepackage{multicol}
\usepackage{float}
\usepackage{tikz}
\usetikzlibrary{arrows.meta, automata, bending, positioning, shapes.misc, patterns, patterns.meta}
\tikzstyle{automaton}=[shorten >=1pt, >={Stealth[bend,round]}, initial text=]
\tikzstyle{accepting}=[accepting below, accepting by double]
\tikzset{
double distance = 1pt,
node distance = 2cm,
every state/.append style={
    minimum size = 0.4cm,
    },
}
\pgfplotsset{width=9cm,compat=1.18}

\usepackage{hyperref}
\hypersetup{
pdfborder = {0 0 0},
}

\newtheorem{thm}{Theorem}
\newtheorem*{teo*}{Theorem}\newtheorem{prop}{Proposition}
\newtheorem{lem}[prop]{Lemma}

\newtheorem{define}{Definition}
\newtheorem{ej}{Example}
\newtheorem*{ej*}{Example}

\newtheorem*{rem}{Comment}

\newtheorem*{note}{Notation}

\newcommand{\HLink}[2]{\hyperref[#2]{#1~\ref*{#2}}}

\providecommand{\thmlabel}[1]{}
\renewcommand{\thmlabel}[1]{\label{thm:#1}}
\newcommand{\thmref}[1]{\HLink{Theorem}{thm:#1}}

\providecommand{\proplabel}[1]{}
\renewcommand{\proplabel}[1]{\label{prop:#1}}
\newcommand{\propref}[1]{\HLink{Proposition}{prop:#1}}

\providecommand{\lemlabel}[1]{}
\renewcommand{\lemlabel}[1]{\label{lem:#1}}
\newcommand{\lemref}[1]{\HLink{Lemma}{lem:#1}}

\providecommand{\deflabel}[1]{}
\renewcommand{\deflabel}[1]{\label{def:#1}}
\newcommand{\defref}[1]{\HLink {Definition}{def:#1}}

\providecommand{\ejlabel}[1]{}
\renewcommand{\ejlabel}[1]{\label{ex:#1}}
\newcommand{\ejref}[1]{\HLink{Example}{ex:#1}}

\providecommand{\corlabel}[1]{}
\renewcommand{\corlabel}[1]{\label{cor:#1}}

\newcommand{\bigo}[0]{\mathcal{O}}

 \everymath{\displaystyle}

\newcommand{\diples}[1]{\left< #1 \right>}

\newcommand{\s}{\stackrel{\sim}{s}}

\begin{document}

\title{Automata for the commutative closure of regular sets}
\author{Verónica Becher \and Simón Lew Deveali \and Ignacio Mollo}
\maketitle
\begin{abstract}
 Consider  \( A^* \), the free monoid generated by the finite alphabet $A$ with the concatenation operation. Two words have the same commutative image when one is a permutation of the symbols of the other. The commutative closure of a set \( L \subseteq A^* \) is the set \( \EuScript{C}(L) \subseteq A^* \) of words whose commutative image coincides with that of some word in \( L \). We provide an algorithm that, given a regular set \( L \), produces a finite state automaton that accepts the commutative closure \( \EuScript{C}(L) \), provided that this closure set is regular. The problem of deciding whether \( \EuScript{C}(L) \) is regular was solved by Ginsburg and Spanier in 1966 using the decidability of Presburger sentences, and by Gohon in 1985 via  formal power series.
The problem of constructing an automaton that accepts \( \EuScript{C}(L) \) has already been studied in the literature. We give a simpler algorithm using an algebraic approach.
\end{abstract}


\section{Primary definitions and Statement of the main result}

An alphabet is a finite, non-empty set of symbols denoted by  $A$. We call the elements of $A$ letters, and the finite sequences of letters words.
We write 
$A^*$ for the set of all words over the alphabet 
$A$.
A language of $A^*$ is any set of words written with letters of $A$. 

A finite state automaton $\mathcal{A}$ is specified by a tuple  $\diples{ Q, A, E, I, T}$ 
where  $Q$ is the set of states,  $A$ is the alphabet,
$E \subseteq Q \times A \times Q$ is the transition set,
$I$ is the set of initial states and $T$ is the set of final states.
A finite state automaton $\mathcal{A}$ = $\diples{ Q, A, E, I, T}$ is deterministic if there is exactly one initial state, $|I|=1$
and for all $p \in Q$ and for all $a\in A$ there exists at most one $q\in Q$ such that $(p,a,q)$ is in $E$.

As usual, we say that a word in $A^*$ is accepted by $\mathcal{A}$ if it is the label of a computation in $\mathcal{A}$ starting from an initial state and ending in a final state.
The language accepted by the finite state automaton $\mathcal{A}$, denoted $L(\mathcal{A})$, is the set of all words accepted by $\mathcal{A}$.
A regular expression over the alphabet $A$ is a formula obtained inductively from the elements of $A$ and the symbols in 
$\{\emptyset, \cup,\lambda,\cdot,*,(, )\}$ 
as follows:
$\emptyset$, $\lambda$, and any letter in $A$ are regular expressions;
If $E$ and $F$ are regular expressions, then $(E \cup F)$, $(E \cdot F)$, and $(E^*)$ are also regular expressions.
We use the standard definition for the language denoted by $E$, see~\cite{Sakarovitch}.
For ease of reading, we simply write $E$ to refer to the language denoted by $E$.

A language is regular if there exists a finite stateautomaton accepting it.
Equivalently, if there exists a regular expression that denotes it.

\begin{note}
We write
$|w|_{a}$ to denote the number of occurrences of the letter $a$ in the word~$w$.
For example $|010|_0=2$.
\end{note}

The commutative image of a word is the number of occurrences of each alphabet letter in the word. It is  known as the Parikh morphism.
Formally, given an alphabet $A$ of $k$ letters, $A =\{a_{1},a_{2},...,a_{k}\}$, the commutative image $\varphi:A^*\to \mathbb {N} ^{k}$ of a word $w\in A^*$ is defined as
\[\varphi(w)=(|w|_{a_{1}},|w|_{a_{2}},\ldots ,|w|_{a_{k}})
\]
The commutative image of a language $L\subseteq A^*$ is $\varphi(L)\subseteq \mathbb {N}^k$,
$\varphi(L)=\{\varphi(w): w\in L\}.$
The commutative closure of a language is defined as follows.
\begin{define}[Commutative Closure]
Let $A$ be an alphabet of $k$ letters and let $\varphi$ be the commutative image of the words over $A$.
If $L\subseteq A^*$,
the commutative closure of $L$, ${\EuScript C}(L)\subseteq A^*$,
\begin{align*}
{\EuScript C}(L) &= \{w\in A^*\ :\ \varphi(w) \in \varphi(L)\}
= \varphi^{-1}(\varphi(L)).
\end{align*}
\end{define}


\begin{rem}
The commutative closure of a regular language is not always regular.
For example, ${\EuScript C}({(ab)^*})=\{w\in A^* \ :\ |w|_a=|w|_b\}$.
\end{rem}

Whether the commutative closure of a language is regular was addressed by Ginsburg and Spanier in 1966~\cite{GINSBURG} and by Gohon in 1985~\cite{GOHON}. The first proof is indirect and relies on the decidability of Presburger sentences. The second provides a simpler algorithm based on formal power series and a result from Eilenberg and Schützenberger~\cite{Eilenberg}. 
We describe this algorithm in the next section.
\begin{prop}[\cite{GINSBURG,GOHON}]
It is decidable to determine whether the commutative closure of a regular language of $A^*$ is regular.
\end{prop}

Once we know that the
commutative closure of a regular language in $A^*$ is regular, 
how can we construct the finite state automaton that recognises it?
The following theorem answers this question and is our main result.

\begin{thm}
\thmlabel{main}
Given a regular expression of a language $L$ in $A^*$ whose commutative closure ${\EuScript C}(L)$ is regular, our algorithm constructs a finite state deterministic automaton for its commutative closure.
\end{thm}

\begin{ej}[Regular commutative closure]
\ejlabel{even}
The commutative closure of the language denoted by $b(a^2\cup b^2)^*$ is regular. The language consists of words with an even number of occurrences of the letter $a$ and an odd number of occurrences of the letter $b$.
Our algorithm produces the finite state automaton in Figure~\ref{fig:automatapares}.

\begin{figure}[t]
\centering
\begin{tikzpicture}[automaton, auto, node distance = 2.5cm]

\node[state, initial] (00) []{};
\node[state] (01) [right of=00]{};
\node[state, accepting] (10) [below of=00]{};
\node[state] (11) [below of=01]{};

\path[->]
 (00) edge [bend left] node [above] { $a$ } (01)
 (01) edge [bend left] node [below] { $a$ } (00)
 (10) edge [bend left] node [above] { $a$ } (11)
 (11) edge [bend left] node [below] { $a$ } (10)
 (00) edge [bend left] node [right] { $b$ } (10)
 (10) edge [bend left] node [left] { $b$ } (00)
 (01) edge [bend left] node [right] { $b$ } (11)
 (11) edge [bend left] node [left] { $b$ } (01);
\end{tikzpicture}
\caption{Finite state automaton accepting ${\EuScript C}({b(a^2\cup b^2)}^*)$.}
\label{fig:automatapares}
\end{figure}
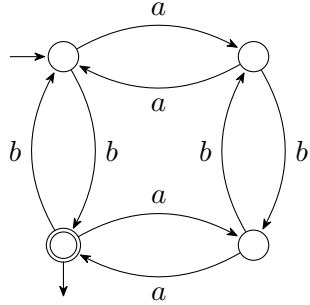
\end{ej}

The problem of constructing a finite state automaton that accepts the commutative closure of a regular language has already been studied by Hoffmann in~\cite{HOFFMANN}.
In his work, Hoffmann provides, in the specific case of group languages, an asymptotic upper bound for the number of states of the resulting automaton, expressed in terms of the number of states of the original automaton.
In \propref{complexity-state}, we present an upper bound for the general case of regular languages with a regular commutative closure. Specifically, we estimate the number of states in the constructed finite state automaton based on the length of the rational expression representing the commutative image of the given regular set.
As a result, the two bounds—Hoffmann's and the one presented here—are not directly comparable in the cases he addressed.

Besides, in Proposition~\ref{prop:complexity-time} we give an upper bound on the number of operations required for our construction, starting from the rational expression of the commutative image of the given regular set.

\section{Gohon's decision algorithm}
Gohon's algorithm~\cite{GOHON} decides whether the commutative closure of a regular language of $A^*$ is regular. Instead of starting from the language, it starts directly from the commutative image of the language. That is, this decision algorithm works directly on $\mathbb{N}^k$. 
When dealing with the commutative closure of sets and when transforming expressions denoting sets, an algebraic approach is both appropriate and useful.

A monoid is a set equipped with an associative binary operation and a neutral element for that operation.
The set $A^*$ is the free monoid over a finite alphabet $A$ with the concatenation as the monoid operation and the empty word $\lambda$  acting as the unit element.
The free commutative monoid generated by a set \( A \) is the quotient of \( A^* \) (the free monoid on \( A \)) by the congruence defined by the relations \( ab = ba \) for all letters \( a \) and \( b \) in \( A \). This commutative monoid is denoted by \( A^\oplus \), where the monoid operation is written as \( + \), and the neutral element is written as \( 0 \).
The monoid $A^\oplus$ is isomorphic to $\mathbb{N}^k$, whose elements are the $k$-tuples of natural numbers $\sigma = (s_1, s_2, \dots , s_k)$. It is also isomorphic to $a^*_1 \times a^*_2 \times \cdots \times a^*_k$ , where $a_1,\dots,a_k$ are the letters of $A$.

Regular expressions are generalised to the setting of ${\mathbb N}^k$
and referred to as rational expressions.
These are defined as one would expect. 

\begin{define}[Rational expression]
A rational expression over $\mathbb{N}^k$ is a formula obtained inductively from elements of $\mathbb{N}^k$ and the 
symbols in $\{\emptyset, \cup, +,\oplus,\left(,\right)\}$ as follows:
$\emptyset$ and every element of $\mathbb{N}^k$ are rational expressions;
If $E$ and $F$ are rational expressions, then $(E \cup F)$, $(E + F)$ and $(E^\oplus)$ are also rational expressions.
\end{define}

The set of   $\mathbb{N}^k$
  denoted by a rational expression 
  is defined similarly to  the set of $A^*$  denoted by a regular expression. 

The set denoted by a rational expression, contained in $\mathbb{N}^k$, is defined in an analogous manner as we defined the set contained in $A^*$ denoted by a regular expression. We simply write $E$ to refer to the set denoted by $E$ and use the standard precedence (first $\oplus$, then $ +$ and last $\cup$).

\begin{define}[Rational set]
A set $S$ of $\mathbb{N}^k$ is rational if
it is denoted by a rational expression over $\mathbb{N}^k$.
\end{define}

The star height of a rational expression is defined for any monoid, we give it just for~${\mathbb N}^k$.

\begin{define}[Star height of a rational expression]
The star height of a rational expression $E$, denoted $h(E)$, is the maximum number of stars nested in the expression $E$. It is defined inductively for rational expressions of $\mathbb{N}^k$:
\begin{align*}
&\text{If }E=\emptyset \text{ or }E \text{ denotes an element of }\mathbb{N}^k \text{ then }
&&h(E)=0.
\\
&\text{If }E=F \cup G\text{ or }E=F + G \text{ then } &&h(E)=\max(h(F),h(G)).
\\
&\text{If }E=F^\oplus \text{ then } && h(E)=1+h(F).
\end{align*}
\end{define}

\begin{ej}
\ejlabel{star}
The star height of ${\big((0,1)^\oplus+(1,0)\big)}^\oplus$ is 2.
\end{ej}

\begin{define}[Star height of a rational set]

The star height of a rational set $S$ of $\mathbb{N}^k$, written $h(S)$, is the minimum star height of the rational expressions of $\mathbb{N}^k$ that denote $S$.
\end{define}

Since the monoid $\mathbb{N}^k$ is commutative, every rational set can be denoted by an expression of star height at most $1$ (see~\cite[Exer. I.6.5]{Sakarovitch}).

\begin{prop}
The star height of a rational set $S\in \mathbb{N}^k$ is at most $1$.
\end{prop}

We therefore define expressions where we restrict the star height.

\begin{define}[Linear expression and semi-linear expression] An expression $E$ denoting a rational set of $\mathbb{N}^k$ is
linear if $E = \gamma+B^\oplus$ where $\gamma \in \mathbb{N}^k$ and $B$ is a finite set of $\mathbb{N}^k$;
it is semi-linear if it is the finite union of linear expressions.
\end{define}
\begin{ej*}[Continuation of \ejref{star}]
${\big((0,1)^\oplus+(1,0)\big)}^\oplus$ is equivalent to the rational expression $(1,0)^\oplus ~\cup~ \Big((1,0)+\big((0,1)\cup(1,0)\big)^\oplus \Big)$ which is semi-linear and has star height $1$.
\end{ej*}

\begin{define}[Non-ambiguous rational operations]
In $\mathbb{N}^k$ we define the non-ambiguous rational operations, a specialization of the rational operations for which we  use the same symbols:

$S\cup T$ is non-ambiguous if $S\cap T = \emptyset$;

$S+T$ is non-ambiguous if 
$\forall s,s' \in S, \forall t,t' \in T$, $( s+t = s'+t' )$ implies $(s=s' \text{ and } t=t')$;

$S^\oplus = \bigcup\limits_{n \in \mathbb{N}}\underbrace{S+\cdots+S}\limits_{n~times}$ is non-ambiguous if each of the sums and unions are non-ambiguous.
\end{define}

 By commutativity of the monoid, if we consider $S=\{s_1,s_2,\ldots,s_l\}\subseteq \mathbb{N}^k$ then
 $$S^\oplus =s_1^\oplus + s_2^\oplus + \cdots + s_l^\oplus = \{n_1s_1 + n_2s_2 + \cdots + n_ls_l \ : \ n_i \in \mathbb{N} \}.$$
 So, $S^\oplus$ is non-ambiguous if and only if whenever 
 $$n_1s_1+\cdots+n_ls_l = m_1s_1+\cdots+m_ls_l$$
 we have $n_i = m_i$ for each $i$.

 \begin{define}[Free basis]

A finite set $B\subseteq \mathbb{N}^k$ is a free basis if
$B^\oplus$ is a non-ambiguous rational expression.
\end{define}

\begin{ej}
$B = \{(1,0),(3,1),(1,1)\}$ is not free because $2(1,0) + (1,1)=(3,1)$. On the other hand $B'=\{(1,0),(1,1)\}$ is a free basis.
\end{ej}

\begin{define}[Simple expression and semi-simple expression] An expression $E$ denoting a set of $\mathbb{N}^k$ is
simple if $E$ is linear ($E = \gamma+B^\oplus$) and $B$ is a free basis;
and it is semi-simple if it is the unambiguous union of simple expressions.
\end{define}

\begin{ej*}[Continuation of \ejref{star}]
The semi-linear rational expression \linebreak 
$(1,0)^\oplus ~\cup~ (1,0)+((0,1)\cup(1,0))^\oplus$ is not semi-simple.
Both $(1,0)^\oplus$ and $(1,0)+((0,1)\cup(1,0))^\oplus$ are simple because both bases are free. However, their union is ambiguous, because their intersection equals $(1,0)+(1,0)^\oplus$.
An equivalent semi-simple expression is
 $$
 (1,0)^\oplus ~\cup~ \Big((1,1)+\big((0,1)\cup(1,0)\big)^\oplus \Big).
 $$
 These expressions are  depicted in \autoref{fig:diagonal graph}.

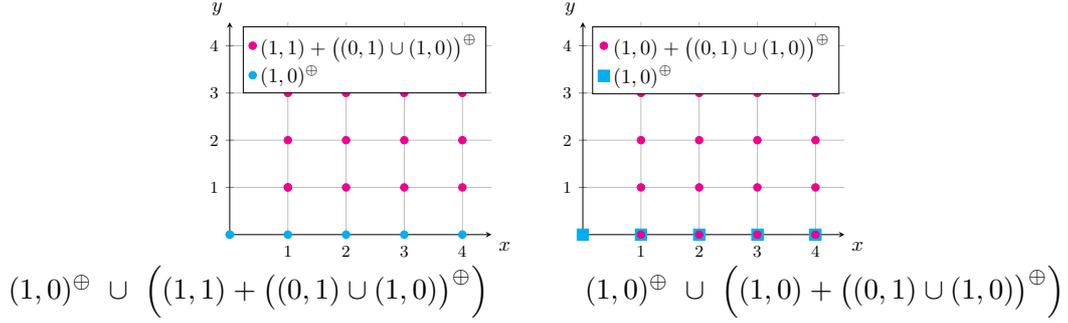
\begin{figure}[t]
\vspace*{-1cm}
 \centering

\usetikzlibrary{patterns}
\pgfplotsset{small}
\begin{tikzpicture}[scale=0.7]
\begin{axis}[
 axis x line=center,
 axis y line=center,
 xmin=0,
 ymin=0,
 xmax=4.5,
 ymax=4.5,
 xlabel={$x$},
 xlabel style={below right},
 ylabel={$y$},
 ylabel style={above left},
 legend columns=3,
 legend cell align=left,
 grid=major,
 transpose legend
 ]

 \addplot [only marks, magenta] coordinates {(1,1)};
 \addlegendentry{$(1,1)+\big((0,1)\cup(1,0)\big)^\oplus$}; \addplot [only marks, cyan] coordinates {(0,0)};
 \addlegendentry{$(1,0)^\oplus$};
 \foreach \i in {0,...,6} {
 \addplot [only marks, cyan] coordinates {(\i,0)};
 }

 \foreach \i in {1,...,6} {
 \foreach \j in {1,...,6}
 \addplot [only marks, magenta] coordinates {(\i,\j)};
 }
\end{axis}
\end{tikzpicture}
~~
\begin{tikzpicture}[scale=0.7]
\begin{axis}[
 axis x line=center,
 axis y line=center,
 xmin=0,
 ymin=0,
 xmax=4.5,
 ymax=4.5,
 xlabel={$x$},
 xlabel style={below right},
 ylabel={$y$},
 ylabel style={above left},
 legend columns=3,
 legend cell align=left,
 grid=major,
 transpose legend
 ]

 \addplot [only marks, mark=*, magenta] coordinates {(1,0)};
 \addlegendentry{$(1,0)+\big((0,1)\cup(1,0)\big)^\oplus$}; \addplot [only marks, mark=square*, mark size=3pt, cyan] coordinates {(0,0)};
 \addlegendentry{${(1,0)}^\oplus$};

 \foreach \i in {1,...,6} {
 \addplot [only marks, mark=square*, mark size=3pt, cyan] coordinates {(\i,0)};
 }

 \foreach \i in {1,...,6} {
 \foreach \j in {0,...,6}
 \addplot [only marks, magenta] coordinates {(\i,\j)};
 }

\end{axis}
\end{tikzpicture}
\begin{tabular}{cccc}
 $(1,0)^\oplus ~\cup~ \Big( (1,1)+\big((0,1)\cup(1,0)\big)^\oplus \Big)$&~&~&$(1,0)^\oplus ~\cup~ \Big((1,0)+\big((0,1)\cup(1,0)\big)^\oplus \Big)$
\end{tabular}
\caption{Semi-simple and semi-linear expressions denoting the same set}
\label{fig:diagonal graph}
\end{figure}
\end{ej*}

\begin{prop}[Eilenberg and  Schützenberger
{\cite[Theorem 4]{Eilenberg}}]
\proplabel{semisimple}
Any rational set of $\mathbb{N}^k$ can
be denoted by a semi-simple expression.
\end{prop}
\begin{rem}
    For 
   the effectiveness of \propref{semisimple}, see \propref{desambiguar} and the results of Chistikov and Hasse in~\cite{HASSE}.
\end{rem}

Finally, we consider formal series on $k$ commutative variables $x_1,\dots,x_k$. In particular, we are interested in the characteristic series of sets of $\mathbb{N}^k$.

\begin{define}\deflabel{charseries}
For every rational subset $S$ of ${\mathbb N}^k$ we denote by $\underline{S}$ the characteristic  series of $S$, defined as follows:
For each element $\sigma\in {\mathbb N}^k$, if $\sigma=(s_1,\dots,s_k)\in S$ then   the coefficient of $x_1^{s_1}\cdots x_k^{s_k}$ in $\underline{S}$ is $1$, 
and it is $0$ otherwise.
\end{define}

\begin{note}
Given a formal series $T$ over $k$ commutative variables $x_1,\dots,x_k$ and an element $\sigma=(s_1,\dots,s_k)\in \mathbb{N}^k$ we write $T\left[(s_1,\dots,s_k)\right]$ to denote the coefficient of $x_1^{s_1}\cdots x_k^{s_k}$ in $T$.
\end{note}

\begin{rem}
    The characteristic series $\underline{S}$ of a rational subset $S\subseteq\mathbb{N}^k$ is unique and completely characterizes $S$.
\end{rem}

\begin{prop}[{\cite[Proposition 3.1]{GOHON}}]
\proplabel{serie}
 Let $S$ be a rational set of $\mathbb{N}^k$. 
  Its characteristic series $\underline{S}$ can be computed recursively from any semi-simple expression of $S$, as follows: \\
  Let $\gamma = (c_1, c_2,\ldots, c_k)$ an element of $\mathbb{N}^k$, $E$ and $F$ subexpressions, and $B=\{\beta_1,\beta_2,\ldots,\beta_l\}$ a free basis. Then:
\begin{eqnarray*}
 \underline{\gamma}&=&x_1^{c_1}x_2^{c_2}...x_k^{c_k}
 \\
 \underline{E + F} &=& \underline{E}\ \underline{F}
 \\
 \underline{E\cup F} &=&\underline{E}+\underline{F}
 \\
 \underline{B^\oplus}&=&\frac{1}{(1-\underline{\beta_1})(1-\underline{\beta_2})\cdots(1-\underline{\beta_l})}.
 \end{eqnarray*}
\end{prop}

\begin{rem}
    Note that in \propref{serie} it is essential to consider a semi-simple expression. For instance, the rule $\underline{E\cup F} =\underline{E}+\underline{F}$ would fail to produce coefficients smaller than $1$ in the formal series for an ambiguous union in the expression.
\end{rem}

\begin{ej}
\ejlabel{Sa}
Let $U$ the set denoted by $(1,1)+(1,1)^\oplus$, and depicted in Figure~\ref{fig:grafodiagonal}. Then,
$\underline{U} = \frac{xy}{(1-xy)}$.
\end{ej}

\begin{prop}[{\cite[Proposition 3.3]{GOHON}}]
\proplabel{unique}
Let $S$ be a rational subset of $\mathbb{N}^k$. There exists a unique pair of polynomials $P(S)$ and $Q(S)$ $\in \mathbb{Z}[x_1,x_2,...,x_k]$ satisfying the following conditions:
\begin{enumerate}
\item[] $\underline{S}=P(S)/Q(S)$.
\item[] $P(S)/Q(S)$ is irreducible.
\item[] $Q(S)[(0,\ldots,0)]= 1$; that is, its constant term is $1$.
\end{enumerate}
Furthermore, for every pair of polynomials $P$ and $Q$ $\in \mathbb{Z}[x_1,x_2,...,x_k]$ such that $\underline{S} = P/Q$, there exists
a polynomial $R\in \mathbb{Z}[x_1,x_2,...,x_k]$ such that $P = R. P(S)$ and $Q = R. Q(S)$.
\end{prop}

\begin{rem}
\propref{unique} gives a canonical representation of the rational subsets of $\mathbb{N}^k$ and is fundamental for proving \propref{gohon}, Gohon's main result in~\cite{GOHON}. However, \propref{unique} is false for non-commutative monoids such as $A^\star$.
\end{rem}

As we already mentioned, not every regular language $L\subseteq A^*$ has a regular commutative closure ${\EuScript C}(L)$, as a subset of $A^*$. 
However, the commutative image of every regular set of $A^*$ is always a rational set of ${\mathbb N}^k$.

\begin{rem}
The commutative image $\varphi(L)$ of any regular language $L\subseteq A^*$ is a rational set of ${\mathbb N}^k$.
\end{rem}

On the other hand, the fact that a set $S$ is rational in $\mathbb{N}^k$ says nothing about whether $\varphi^{-1}(S)$ is a regular set of $A^*$ or not. 
A well studied subclass of the rational sets   in {finitely generated} monoids is the class of the recognizable sets~\cite{Sakarovitch}. 
Our interest in these sets  comes from the fact that, for any surjective morphism,  the inverse image of a recognizable set of ${\mathbb N}^k$ is a regular set of $A^*$.

\begin{define}[Recognizable set]
Let $A$ be a finite alphabet and let
$\varphi:A^* \to \mathbb{N}^k$ be the commutative image.
A subset $S$ of $\mathbb{N}^k$ is recognizable if
$\varphi^{-1}(S)$ is regular.
\end{define}

Note that if $S$ is recognizable in $\mathbb{N}^k$ then $\varphi(\varphi^{-1}S)=S$ is rational and thus\linebreak $\text{Rec}(\mathbb{N}^k)\subseteq\text{Rat}(\mathbb{N}^k)$.
The following is an equivalent characterization of recognizable sets that is attributed to Mezei.

\begin{define}[{Mezei~\cite[Corollary II.2.20]{Sakarovitch}}]
\deflabel{mezei}
A set $S$ of $\mathbb{N}^k$ is recognizable if
there exists a family of sets
$\{T_{i,j}\}_{i\in I,1\leq j\leq k}$ with $I$ finite and each $T_{i,j}$  rational subsets of  $\mathbb{N}$
such that
$$
S=\bigcup\limits_{i \in I} T_{i,1} \times \cdots \times T_{i,k}.
$$
\end{define}

\begin{rem}
The characterization of recognizable sets as in \defref{mezei} is not unique.
For instance, when $S$ is $\mathbb{N}$, we already find multiple characterizations:
$1^\oplus$; 
 $1 \cup (1 + 1^\oplus)$;
 $1 \cup \cdots \cup n \cup \left((n+1) + 1^\oplus\right)$;
 $2^\oplus \cup ( 1+2^\oplus)$;  \ldots
\end{rem}

The next result is the main  theorem obtained by Gohon in~\cite{GOHON}.

\begin{prop}[{{Gohon~\cite[Theorem 4.6]{GOHON}}}]
\proplabel{gohon}
For every rational  set $S $ of ${\mathbb{N}^k}$, we can associate a fraction $P/Q=\underline{S}$, where $P$ and $Q$ are polynomials of $k$ commutative variables $x_1,x_2,...,x_k$ with coefficients in $\mathbb{Z}$, such that $S$ is recognizable if and only if $Q=1$, or $Q$ is a product of polynomials of the form $(1-x_j^{p_j})$.
\end{prop}

Gohon's algorithm starts from a semi-simple expression and obtains its characteristic series. 
It reduces the polynomials until all factors of more than one variable in the denominator are simplified. This is possible exactly when the set is recognizable.

\begin{ej*}[Continuation of \ejref{Sa}] 
$\underline{U} = \frac{xy}{(1-xy)}$.
The set $U$, shown in \autoref{fig:grafodiagonal}, is not recognizable because $\underline{U}$ is irreducible and its denominator features a factor with more than one variable.
In fact, 
 if we consider the alphabet $A=\{a,b\}$
the set $U$ corresponds to  ${\EuScript C}({(ab)^+})=\{w\in A^+ \ :\ |w|_a=|w|_b\}={\EuScript C}({(ab)^*})-\{\lambda\}$, which is not regular.
\end{ej*}
\begin{figure}[t]
\vspace*{-1cm}
 \centering
\usetikzlibrary{patterns}
\begin{tikzpicture}[scale=0.6]
\begin{axis}[
 axis x line=center,
 axis y line=center,
 xmin=0,
 ymin=0,
 xmax=6.5,
 ymax=6.5,
 xlabel={$x$},
 xlabel style={below right},
 ylabel={$y$},
 ylabel style={above left},
 legend columns=3,
 legend cell align=left,
 grid=major,
 legend style={ at= {(1,-0.1)},},
 transpose legend
 ]

 \foreach \i in {1,...,6} {
 \addplot [only marks, Green] coordinates {(\i,\i)};
 }
\end{axis}
\end{tikzpicture}
\hspace{1cm}
\begin{tikzpicture}[scale=0.6]
\begin{axis}[
 axis x line=center,
 axis y line=center,
 xmin=0,
 ymin=0,
 xmax=6.5,
 ymax=6.5,
 xlabel={$x$},
 xlabel style={below right},
 ylabel={$y$},
 ylabel style={above left},
 legend columns=3,
 legend cell align=left,
 grid=major,
 legend style={ at= {(1,-0.1)},},
 transpose legend
 ]

 \foreach \i in {1,3,...,6} {
 \foreach \j in {0,2,...,6}
 \addplot [only marks, violet] coordinates {(\j,\i)};
 }
\end{axis}
\end{tikzpicture}
\begin{tabular}{lcr}
 $U=(1,1)+(1,1)^\oplus$&\hspace*{2.4cm}&$V=(0,1)+\big((2,0)\cup(0,2)\big)^\oplus$
\end{tabular}

\caption{On the left, an infinite  set of ${\mathbb N}^2$ that is not recognizable, because it  is impossible to characterize it as a finite union of products of rational sets in $\mathbb{N}$. On the right, an infinite   set of ${\mathbb N}^2$ that is recognizable: it can be described as the product between the even and the odd numbers, both rational sets  in $\mathbb{N}$.}
\label{fig:grafodiagonal}
 \label{fig:garfopares}
\end{figure}
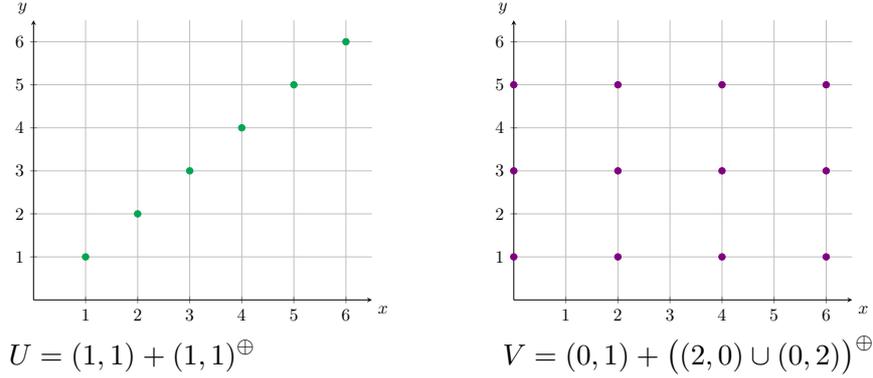
\begin{ej}
$V = (0,1)+\big((2,0)\cup(0,2)\big)^\oplus$,
 $\underline{V} = \frac{y}{(1-x^2)(1-y^2)}$.
The set $V$,  shown in \autoref{fig:garfopares}, is recognizable because  $\underline{V}$ has no factors of more than one variable in the denominator.
Notice that $V$ is the commutative image of $L=b(a^2\cup b^2)^*$
in \ejref{even} and we have already  given an  automaton that accepts ${\EuScript C}({L})$; thus, $V$ is recognizable.
\end{ej}

\section{Resimple expressions}
\label{resimple}

We introduce a new kind of rational expression, that we call resimple expressions,  which are relatively simple and  denote  recognizable sets of 
$\mathbb{N}^k$. Our goal is to construct these resimple expressions from the characteristic series of recognizable sets.

\begin{note}
From now on, let 
$\varphi:A^\star\rightarrow\mathbb{N}^k$
denote the commutative image morphism.\linebreak
We write $\mathbf{e}_i$ to denote the element of $\mathbb{N}^k$ whose components are all $0$ except for the $i$-th, which is~$1$. The set $\{\mathbf{e}_1,\dots,\mathbf{e}_k \}$ is the set of generators of~$\mathbb{N}^k$.
\end{note}

\begin{define}[Primary element]
An element $\sigma \in \mathbb{N}^k$ is {primary} if $\sigma=n \mathbf{e}_j$ for some $n \in \mathbb{N}_{>0}$, and some  $j$ between $1$ and $k$.
When we want to make explicit the index $j$ we  say that $\sigma$ is $j$-primary.
\end{define}

\begin{define}[Primary basis]
A free basis $B\subseteq {\mathbb N}^k$ is {primary} if all its elements are primary.
\end{define}

We use the name atomic resimple for simple expressions that denote  recognizable  sets.

\begin{define}[Atomic resimple expression]
An expression $R$ of a recognizable set of $\mathbb{N}^k$ is {atomic resimple} 
if $R$ is simple of the form  $R = \gamma+B^\oplus$ where  $B$ is a free primary basis.
\end{define}

\begin{prop}
\proplabel{atomic resimple denote recognizable}
    If $R=(c_1,\dots,c_k)+B^\oplus$ is an atomic resimple expression then it denotes a recognizable set.
\end{prop}
\begin{proof}
    $B$ is a free primary base, so all its elements are primary. Also, because it is free, there is at most one $j$-primary element in $B$ for each $j$. 
    Let us consider $S_j$ a subset of $\mathbb{N}$ 
    $$
S_j = \begin{cases}
c_j+p_j^\oplus  & \text{ if }  p_j \mathbf{e}_j \in B\\
{c_j} & \text{ otherwise}.
\end{cases}
$$
    
    Then the set denoted by $R$ is exactly $S_1\times\cdots \times S_k$, a recognizable set.
\end{proof}

An atomic resimple expression is a rational expression of $\mathbb{N}^k$. 
We just write $R$ for the set denoted by an atomic resimple expression $R$.
Starting from an atomic resimple expression $R = c + p^\oplus$, where $c,p\in \mathbb{N}$, we can construct a finite state deterministic automaton $\mathcal{A} = \diples{ Q, \{a\}, E, I, T}$ that accepts the language $\varphi^{-1}(R) \subseteq {A}^*$, where $\varphi$ is the commutative image. The automaton is depicted in Figure~\ref{fig:automatacaracol}. 
We must do the same for any ~$k$.
To accomplish this we must first return to $A^*$ to introduce an operation between regular languages.

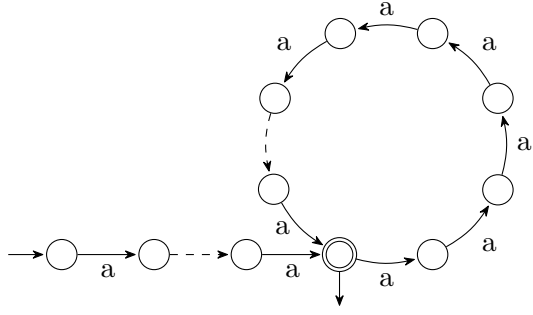
\begin{figure}[t]
\vspace*{-1cm}
 \centering
\begin{tikzpicture}[automaton, auto,
node distance = 0.8cm]

 \node[state, initial] (B){};
 \node[state] (C) [right =of B]{};
 \node[state] (D) [right =of C]{};
 \node[state, accepting] (G) [right =of D]{};
 \node[state] (H) [right =of G]{};
 \node[state] (I) [above right =of H]{};
 \node[state] (J) [above =of I] {};
 \node[state] (K) [above left =of J]{};
 \node[state] (L) [left =of K]{};
 \node[state] (M) [below left =of L]{};
 \node[state] (N) [above left =of G]{};
 \path[->] (B) edge node [below] {a} (C);
 \path[->] (D) edge node [below] {a} (G);
 \path[->] (G) edge [bend right = 15] node [below] {a} (H);
 \path[->] (H) edge [bend right = 15] node [below right] {a} (I);
 \path[->] (I) edge [bend right = 15] node [right] {a} (J);
 \path[->] (J) edge [bend right = 15] node [above right] {a} (K);
 \path[->] (K) edge [bend right = 15] node [above] {a} (L);
 \path[->] (L) edge [bend right = 15] node [above left] {a} (M);
 \path[->] (N) edge [bend right = 15] node [left] {a} (G);
 \path[->] (C) edge [dashed] (D);
 \path[->] (M) edge [bend right = 15] [dashed] (N);
\end{tikzpicture}
\caption{
Finite state automaton for $\varphi^{-1}(R)$ when $R$ is an atomic resimple expression of $\mathbb{N}$ ($k=1$ and $A=\{a\}$).
}
\label{fig:automatacaracol}
\end{figure}

\begin{define}[Shuffle]
The shuffle of two words $w$ and $v$ in $A^*$, denoted with $w \between v$, is the subset of $A^*$ defined by
$$
w \between v=\{w_1v_1\cdots w_nv_n\ : \
w = w_1\cdots w_n, v = v_1\cdots v_n, \text{with }
w_i,v_j \in A^*\
\text{for each $i,j$ and $n \in \mathbb{N}$}
\}
$$
The shuffle of two words is additively extended to the shuffle of two languages in a natural way. If $L,K\subseteq A^\star$ are two languages, their shuffle is defined as follows:
$$
L \between K= \bigcup\limits_{w\in L, v \in K}w \between v.
$$
\end{define}

\begin{ej}
$ab \between ba = \{abba, baba, baab, abab\}$.
\end{ej}

The shuffle of languages is an associative operation on $\mathcal{P}(A^\star)$. Therefore, if $L_1,\dots,L_l$ are languages, we denote
$$
L_1\between\cdots\between L_l = \bigcup\limits_{w_i \in L_i}w_1 \between \cdots \between w_l.
$$

We introduce the shuffle  product automaton.

\begin{define}[Shuffle product automaton]
Let $\mathcal{A'} = \diples{ Q', A, E', I', T' }$ and \linebreak $\mathcal{A''} = \diples{ Q'', A, E'', I'', T'' }$ be deterministic finite state automata.
We define the shuffle product automaton $\mathcal{A}=\mathcal{A'}\between \mathcal{A''}$ as $\mathcal{A} = \diples{Q' \times Q'', A, E, I' \times I'', T' \times T''}$,
where\begin{align*}
E=&\{((p',p''),a,(q',p'')): p''\in Q'' \text{ and }((p',a,q')\in E'\} \cup
\\
&\{((p',p''),a,(p',q'')): p'\in Q' \text{ and }((p'',a,q'')\in E''\}.
\end{align*}
Inductively, given finite state deterministic automata $\mathcal{A}_1,\dots,\mathcal{A}_n$, we define  an automaton $\mathcal{A}=\mathcal{A}_1\between\dots\between\mathcal{A}_n$ by repeating this procedure.
\end{define}

It is easy to see that $L(\mathcal{A'}\between \mathcal{A''})=L(\mathcal{A'})\between L(\mathcal{A''})$. And, furthermore, that \linebreak
$L(\mathcal{A}_1\between\cdots\between \mathcal{A}_n)=L(\mathcal{A}_1)\between\cdots\between L(\mathcal{A}_n)$.

\begin{prop}
The shuffle of two regular languages in $A^*$ is a regular language.
\end{prop}

\begin{prop}
Let $A'$ and $A''$ be  disjoint alphabets, and let $A=A'\cup A''$ be  the union of both. Suppose we have two finite state deterministic   automata $\mathcal{A'}$ and $\mathcal{A''}$ defined on alphabets $A'$ and $A''$ respectively. If we consider both as  automata on $A$, then the shuffle between them $\mathcal{A'}\between\mathcal{A''}$ is also a deterministic finite state automaton.
\end{prop}

\begin{proof}
Assume, forcontradiction, that $\mathcal{A'}\between\mathcal{A''}$ is not deterministic.
 Since $\mathcal{A'}$ and $\mathcal{A''}$ are deterministic, both have only one initial state, so $\mathcal{A'}\between\mathcal{A''}$  has only one initial state too.
Then there must be a state $(p',p'')$ and a letter $a\in A$ such that there's more than one transition labelled $a$ leaving $(p',p'')$. 
Since the alphabets are disjoint there are two exclusive possibilities: either $a\in A'$ or $a\in A''$. Without loss of generality we assume $a\in A'$.
As there are no transitions in $E''$ labelled with $a$, both must be in~$E'$. So there is more than one transition in $E'$ with origin $p'$ and label $a$. This is impossible because $\mathcal{A}'$ is deterministic.
\end{proof}

\begin{note} In the sequel we  write automaton as an abbreviation of  finite state deterministic automaton.
\end{note}

Now we build the automaton over $A^*$ for any atomic resimple expression,  as we did for~$k=1$.

\begin{prop}
\proplabel{resimple automaton}
Let $R=(c_1,\ldots,c_k)+{B}^\oplus$ be an atomic resimple expression denoting a set $S \subseteq \mathbb{N}^k$. We can build an automaton $\mathcal{A}$ that accepts the language $\varphi^{-1}(S) \subseteq A^*$, where $\varphi$ is the commutative image.
\end{prop}

\begin{proof}
By \propref{atomic resimple denote recognizable} we know that $S$ can be written in the form
$S_1 \times \cdots \times S_k$, where each $S_j$ is denoted by a resimple expression $S_j=c_j+p_j^\oplus$. For each coordinate $j$ we can obtain a complete  automaton $\mathcal{A}_j$ labelled over the alphabet $\{a_j\}$ that accepts $\varphi^{-1}(S_j)$, see \autoref{fig:automatacaracol}.
Consider the automaton $\mathcal{A}_j$ as labelled in the larger alphabet $A=\{a_1,\ldots, a_k\}$. Then the shuffle product
$\mathcal{A} = \mathcal{A}_1 \between\dots\between\mathcal{A}_k$
accepts the language $\varphi^{-1}(S)$.   Observe that $w\in L(\mathcal{A})$, if and only if, for every coordinate $j$ we have  $|w|_{a_j}\in S_j$, which is equivalent to $\varphi(w)\in S$.
\end{proof}

\usetikzlibrary{patterns}
\begin{figure}[t]
\begin{subfigure}{.5\textwidth}
 \centering
\begin{tikzpicture}[scale=0.6]
\begin{axis}[
 axis x line=center,
 axis y line=center,
 xmin=0,
 ymin=0,
 xmax=6.5,
 ymax=6.5,
 xlabel={$x$},
 xlabel style={below right},
 ylabel={$y$},
 ylabel style={above left},
 legend columns=3,
 legend cell align=left,
 grid=major,
 legend style={ at= {(1,-0.1)},},
 transpose legend
 ]

 \foreach \i in {1,...,6} {
 \foreach \j in {0,...,6}
 \addplot [only marks, orange] coordinates {(\j,\i)};
 }
\end{axis}
\end{tikzpicture}
 \caption{$(0,1)+\big((1,0)\cup(0,1)\big)^\oplus$}
 \label{fig:grafosineje}
\end{subfigure}%
\begin{subfigure}{.5\textwidth}

\centering
\begin{tabular}{m{1.5cm} | m{5cm} }
\huge{\ ${\between}$}
&
\begin{tikzpicture}[automaton, auto]
\node[state, initial] (0){};
\node[state, accepting] (1) [right =of 0]{};
 \path[->]
 (0) edge node [above] { $(0,1)$ } (1)
 (1) edge [loop right] node {$(0,1)$} ();
\end{tikzpicture}\\
\hline
\\
\begin{tikzpicture}[automaton, auto]
\node[state, initial, accepting] (0){};
 \path[->]
 (0) edge [loop above] node {$(1,0)$} ();
\end{tikzpicture}&
\begin{tikzpicture}[automaton, auto]

\node[state, initial] (0){};
\node[state, accepting] (1) [right =of 0]{};
 \path[->]
 (0) edge node [above] { $(0,1)$ } (1)
 (0) edge [loop above] node {$(1,0)$} ()
 (1) edge [loop above] node {$(1,0)$} ()
 (1) edge [loop right] node {$(0,1)$} ();
\end{tikzpicture}
\end{tabular}
\caption{Construction of the  automaton.}
\label{fig:automatasineje}
\end{subfigure}
\caption{Original set and the automata construction of \propref{resimple automaton}.}
\end{figure}
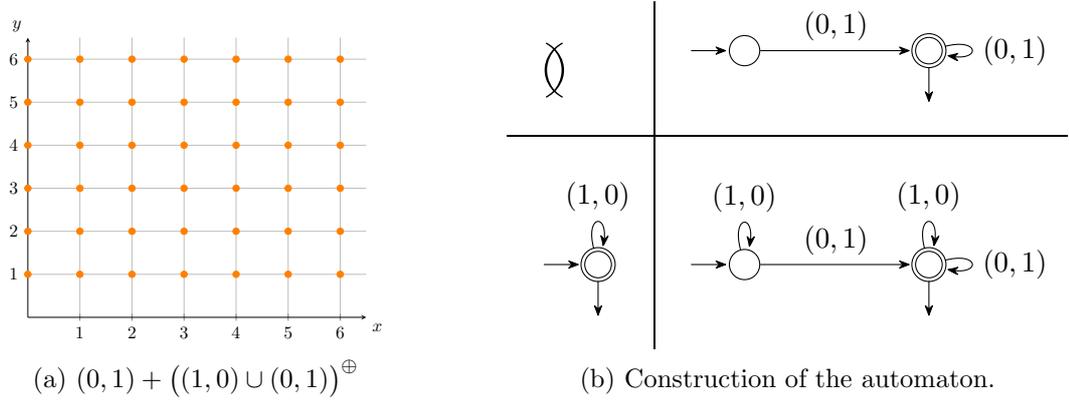

 \begin{rem}
The  automata in the examples accept languages in $A^*$, where $A = \{\mathbf{e}_1, \ldots, \mathbf{e}_k\}$ is the set of the generators of $\mathbb{N}^k$ considered as an alphabet.
For clarity, all examples are for the case~$\mathbb{N}^2$.
From now on, the automaton for a resimple expression $R$ will denote  the  automaton that accepts $\varphi^{-1}(R)$.
\end{rem}
\begin{ej} The expression 
$R=(0,1)+\big((1,0)\cup(0,1)\big)^\oplus$ is  resimple,
see \autoref{fig:grafosineje}.
The automaton for this expression is the shuffle product between the automata for $R_x=1^\oplus$ and the automata for $R_y=1+1^\oplus$, see \autoref{fig:automatasineje}.
\end{ej}

It is possible to calculate the exact size of the automaton built in \propref{resimple automaton}.
\begin{prop}
\proplabel{cotaautomata}
 Let $R=(c_1,\ldots,c_k)+{B}^\oplus$ be an atomic resimple expression. For each coordinate $j$, let $p_j\in {\mathbb N}$ be
$$
p_j = \begin{cases}
n & \text{ if } n \mathbf{e}_j \in B\\
2 & \text{ otherwise}.
\end{cases}
$$
Then, 
the complete automaton 
defined in  \propref{resimple automaton}  accepts $\varphi^{-1}(R)$ and it has exactly 
$\prod\limits_{j=1}^k c_j+p_j$ states.
\end{prop}

\begin{rem}
If there is no $ j$-primary element in $B$ then 
the resimple expression for that component,
 $R_j=\gamma_j$,
denotes a finite set. Therefore, the  automaton for $R_j$ would have $\gamma_j+1$ states. However,  to make it {complete} it is necessary to add a sink state, having $\gamma_j+2$ states in total. This explains $p_j = 2$.
\end{rem}

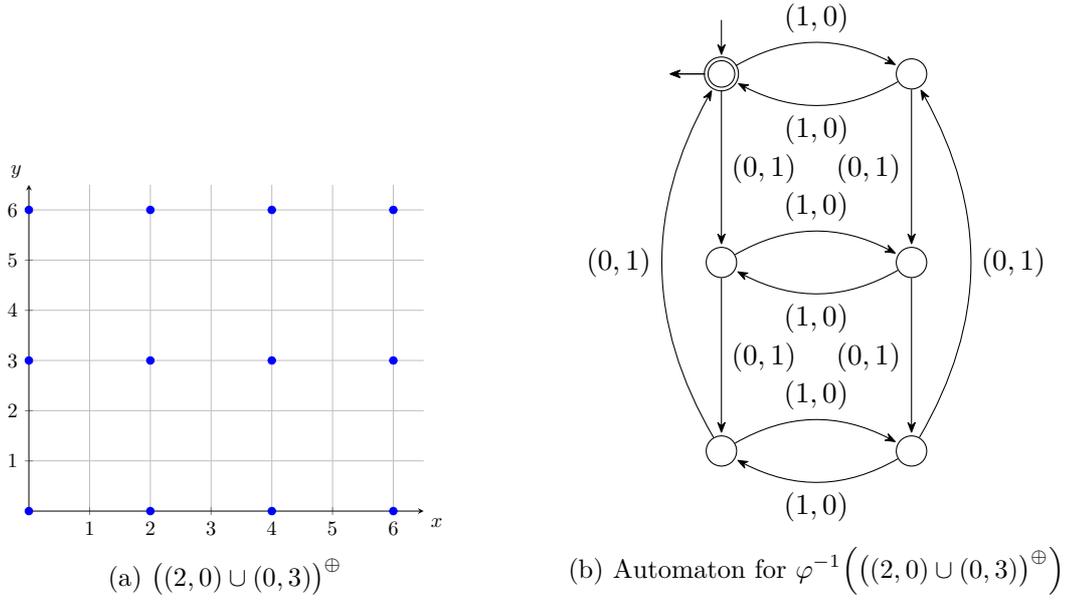
\begin{figure}[t]
\vspace*{-1cm}
\begin{subfigure}{.5\textwidth}
 \centering
\usetikzlibrary{patterns}
\begin{tikzpicture}[scale=0.7]
\begin{axis}[
 axis x line=center,
 axis y line=center,
 xmin=0,
 ymin=0,
 xmax=6.5,
 ymax=6.5,
 xlabel={$x$},
 xlabel style={below right},
 ylabel={$y$},
 ylabel style={above left},
 legend columns=3,
 legend cell align=left,
 grid=major,
 legend style={ at= {(1,-0.1)},},
 transpose legend
 ]

 \foreach \i in {0,2,...,6} {
 \foreach \j in {0,3,...,6}
 \addplot [only marks, blue] coordinates {(\i,\j)};
 }
\end{axis}
\end{tikzpicture}
 \caption{$\big((2,0)\cup(0,3)\big)^\oplus$}
 \label{fig:graph23}
\end{subfigure}%
\begin{subfigure}{.5\textwidth}
 \centering
\begin{tikzpicture}[automaton, auto, node distance = 2.5cm]

\node[state, initial above, accepting left, accepting] (00) []{};
\node[state] (01) [right of=00]{};
\node[state] (10) [below of=00]{};
\node[state] (20) [below of=10]{};
\node[state] (11) [below of=01]{};
\node[state] (21) [below of=11]{};

 \path[->]
 (00) edge [bend left] node [above] { $(1,0)$ } (01)
 (01) edge [bend left] node [below] { $(1,0)$ } (00)
 (10) edge [bend left] node [above] { $(1,0)$ } (11)
 (11) edge [bend left] node [below] { $(1,0)$ } (10)
 (20) edge [bend left] node [above] { $(1,0)$ } (21)
 (21) edge [bend left] node [below] { $(1,0)$ } (20)

 (00) edge node [right] { $(0,1)$ } (10)
 (10) edge node [right] { $(0,1)$ } (20)
 (20) edge [bend left] node [left] { $(0,1)$ } (00)

 (01) edge node [left] { $(0,1)$ } (11)
 (11) edge node [left] { $(0,1)$ } (21)
 (21) edge [bend right] node [right] { $(0,1)$ } (01);
\end{tikzpicture}
 \caption{Automaton for $\varphi^{-1}\Big(\big((2,0)\cup(0,3)\big)^\oplus\Big)$} 
 \end{subfigure}
 \caption{Original set and the resulting automaton.}
  \label{fig:automata23}
\end{figure}

\begin{ej} An atomic resimple expression and its automaton appear  in \autoref{fig:automata23}.
\end{ej}

Yet another instance of an atomic resimple expression can be found in  \ejref{even}.

\begin{define}[Resimple expression of a set of $\mathbb{N}^k$]
\deflabel{resimple}
A resimple expression of a recognizable set of $\mathbb{N}^k$ is a formula that is obtained inductively from atomic resimple expressions and the boolean operations: union, intersection and complement $(\cup,\cap \text{ and } ^c)$:
An atomic resimple expression is resimple.
If $E, F$ are resimple expressions then $E\cup F$, $E\cap F$ and $E^c$ are resimple expressions.
\end{define}

\begin{figure}[t]
 \centering

\usetikzlibrary{patterns}
\begin{tikzpicture}[scale=0.7]
\begin{axis}[
 axis x line=center,
 axis y line=center,
 xmin=0,
 ymin=0,
 xmax=6.5,
 ymax=6.5,
 xlabel={$x$},
 xlabel style={below right},
 ylabel={$y$},
 ylabel style={above left},
 legend columns=3,
 legend cell align=left,
 grid=major,
 legend style={ at= {(1.7,0.2)},},
 transpose legend
 ]

\addplot [only marks, purple] coordinates {(0,0)};
\addlegendentry{$\big((2,0)\cup(0,1)\big)^\oplus$ }

\addplot [ mark=square*,only marks, olive] coordinates {(1,1)};
\addlegendentry{$(1,1)+\big((2,0)\cup(0,1)\big)^\oplus$}

 \foreach \i in {0,2,...,6} {
 \foreach \j in {0,...,6}
 \addplot [only marks, purple] coordinates {(\i,\j)};
 }

 \foreach \i in {1,3,...,6} {
 \foreach \j in {1,...,6}
 \addplot [mark=square*, only marks, olive] coordinates {(\i,\j)};
 }
\end{axis}
\end{tikzpicture}
 \caption{$\Big((0,1)+\big((1,0)\cup(0,1)\big)^\oplus \Big) \cup \Big( \big((2,0)\cup(0,3)\big) ^\oplus \Big)$}
 \label{fig:grafounion}
\end{figure}
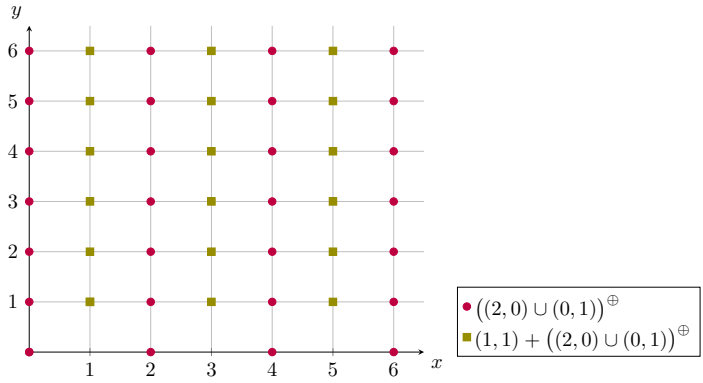

\begin{ej}
\ejlabel{Sc}
$(0,1)+((1,0)\cup(0,1))^\oplus \cup ((2,0)\cup(0,3))^\oplus$ is resimple, see \autoref{fig:grafounion}.
\end{ej}
\begin{prop}
\proplabel{half face}
 Let $S$ be a set of $\mathbb{N}^k$.
 If $S$ has a resimple expression that denotes it then $S $ a recognizable set in  $\mathbb{N}^k)$.
\end{prop}
\begin{proof}
The family of recognizable sets of $\mathbb{N}^k$ forms a Boolean algebra (see~\cite{Sakarovitch}) and by \propref{atomic resimple denote recognizable} we know that atomic resimple expressions denote recognizable sets.
\end{proof}

\begin{define}[Consistent resimple expression]
An expression is resimple {consistent} if it is resimple and all the bases of its atoms are the same.
\end{define}

\begin{ej*}[Continuation of \ejref{Sc}]
The resimple expression $\Big((0,1)+\big((1,0)\cup(0,1)\big)^\oplus \Big) \cup \big((2,0)\cup(0,3)\big)^\oplus$, is not consistent because  $\{(1,0),(0,1)\}\neq\{(2,0),(0,3)\}$.
However, \linebreak
$\big((2,0)\cup(0,1)\big)^\oplus \cup \Big((1,1)+\big((2,0)\cup(0,1)\big)^\oplus \Big)$ is  a consistent resimple expression that denotes the same set.
\end{ej*}

\begin{thm}
\thmlabel{automata}
For any resimple expression denoting a recognizable set
 $S$ in $\mathbb{N}^k$  an automaton can be effectively constructed to recognize the language $\varphi^{-1}(S) \subseteq A^*$, where $\varphi$ is the commutative image.
\end{thm}
\begin{proof}
This is an immediate consequence of the \propref{resimple automaton} and the fact that regular languages in $A^*$ form an effective boolean algebra.
\end{proof}

Propositions~\ref{prop:intersection_is_regular} and~\ref{prop:intersection_is_deterministic} below are standard and can be found in any automata theory book (see for instance~\cite{Sakarovitch}). We give the proof of Proposition~\ref{prop:intersection_is_regular} in detail to introduce the intersection automaton.

\begin{prop}
\label{prop:intersection_is_regular}
The intersection of two regular languages in $A^*$ is regular.
\end{prop}
\begin{proof}
Let $\mathcal{A'} = \diples{ Q', A, E', I', T' }$ and $\mathcal{A''} = \diples{ Q'', A, E'', I'', T'' }$.
We define $\mathcal{A}=\mathcal{A'}\cap\mathcal{A''}$ as $\mathcal{A} = \diples{ Q' \times Q'', A, E, I' \times I'', T' \times T''}$, with 
\begin{align*}
E=&\{((p',p''),a,(q',q'')): (p',a,q')\in E'\text{ and }(p'',a,q'')\in E''\}.
\end{align*}
It is easy to see that $L(\mathcal{A})=L(\mathcal{A'})\cap L(\mathcal{A''})$.
\end{proof}
The automaton $\mathcal{A'}\cap\mathcal{A''}$ that we used in the previous proof is called the \textit{intersection automaton} between the automata $\mathcal{A'}$ and $\mathcal{A''}$ and it is a general construction that allows us to obtain an automaton to recognise the intersection of two languages from two automata that recognise each of those languages. 
\begin{prop}
\label{prop:intersection_is_deterministic}
If $\mathcal{A'}$ and $\mathcal{A''}$ are  deterministic automata, then the intersection automaton $\mathcal{A} = \mathcal{A'}\cap\mathcal{A''}$ is also deterministic.
\end{prop}

We are interested in seeing that the  intersection automaton does not grow significantly.

\begin{prop}
\proplabel{inter}
If $\mathcal{A'}$ and $\mathcal{A''}$ are defined from resimple expressions that are consistent with each other, then the number of states in the intersection automaton $\mathcal{A} = \mathcal{A'} \cap \mathcal{A''}$ for one coordinate is at most the maximum between the number of states in $\mathcal{A'}$ and $\mathcal{A''}$ for that coordinate.
\end{prop}
\begin{proof}
Since in each coordinate the two automata end up with the same period, there is a mapping between the states of the two automata. In each coordinate, the intersection of the two automata is just  the bigger one, see \autoref{fig:pairing},  possibly with different  final states.
\end{proof}

\begin{figure}[t]
\vspace*{-4cm}
\rotatebox{-45}{
\begin{tikzpicture}[automaton, auto,
node distance = 0.8cm]

\node[state, initial] (B){};
\node[state, fill=gray] (C) [above right =of B]{};
\node[state, fill=brown] (D) [above right =of C]{};
\node[state, fill=yellow] (G) [above right =of D]{};
\node[state, fill=orange] (H) [right =of G]{};
\node[state, fill=pink] (I) [above right =of H]{};
 \node[state, fill=red] (J) [above =of I] {};
 \node[state, fill=violet] (K) [above left =of J]{};
 \node[state, fill=blue] (L) [left =of K]{};
 \node[state, fill=cyan] (M) [below left =of L]{};
 \node[state, fill=Green] (N) [above left =of G]{};
 \path[->] (B) edge node [below] {a} (C);
 \path[->] (C) edge node [below] {a} (D);
 \path[->] (D) edge node [below] {a} (G);
 \path[->] (G) edge [bend right = 15] node [below] {a} (H);
 \path[->] (H) edge [bend right = 15] node [below right] {a} (I);
 \path[->] (I) edge [bend right = 15] node [right] {a} (J);
 \path[->] (J) edge [bend right = 15] node [above right] {a} (K);
 \path[->] (K) edge [bend right = 15] node [above] {a} (L);
 \path[->] (L) edge [bend right = 15] node [above left] {a} (M);
 \path[->] (N) edge [bend right = 15] node [left] {a} (G);
 \path[->] (M) edge [bend right = 15] node {a} (N); \node[state, initial] (B){};
 \node[state, fill=gray] (B1) [right = 4.5cm of B]{};
 \node[state, initial] (A1)[below left = of B1]{};
 \node[state, fill=brown] (C1) [above right =of B1]{};
 \node[state, fill=yellow] (D1) [above right =of C1]{};
 \node[state, fill=orange] (G1) [right = 3.5cm of H]{};
 \node[state, fill=pink] (H1) [right = 3.5cm of I]{};
 \node[state, fill=red ] (I1) [right = 3.5cm of J]{};
 \node[state, fill=violet] (J1) [right = 3.5cm of K] {};
 \node[state, fill=blue] (K1) [right = 3.5cm of L]{};
 \node[state, fill=cyan] (L1) [right = 3.5cm of M]{};
 \node[state, fill=Green] (M1) [right = 3.5cm of N]{};
 \node[state, fill=yellow] (N1) [right = 3.5cm of G]{};
 \path[->] (A1) edge node [below] {a} (B1);
 \path[->] (B1) edge node [below] {a} (C1);
 \path[->] (C1) edge node [below] {a} (D1);
 \path[->] (D1) edge node [below] {a} (G1);
 \path[->] (G1) edge [bend right = 15] node [below right] {a} (H1);
 \path[->] (H1) edge [bend right = 15] node [right] {a} (I1);
 \path[->] (I1) edge [bend right = 15] node [above right] {a} (J1);
 \path[->] (J1) edge [bend right = 15] node [above] {a} (K1);
 \path[->] (K1) edge [bend right = 15] node [above left] {a} (L1);
 \path[->] (M1) edge [bend right = 15] node [left] {a} (N1);
\path[->] (N1) edge [bend right = 15] node [below left] {a}  (G1);
\path[->] (L1) edge [bend right = 15] node {a} (M1);
\end{tikzpicture}
}
\vspace*{-3cm}
\caption{The automata  for two resimple consistent expressions with $k=1$. The intersection automaton is equal to  the bigger automaton (right) since all the states of the smaller automaton (left) are paired up with the states of the bigger one (right), color by color.} 
\label{fig:pairing}
\end{figure}
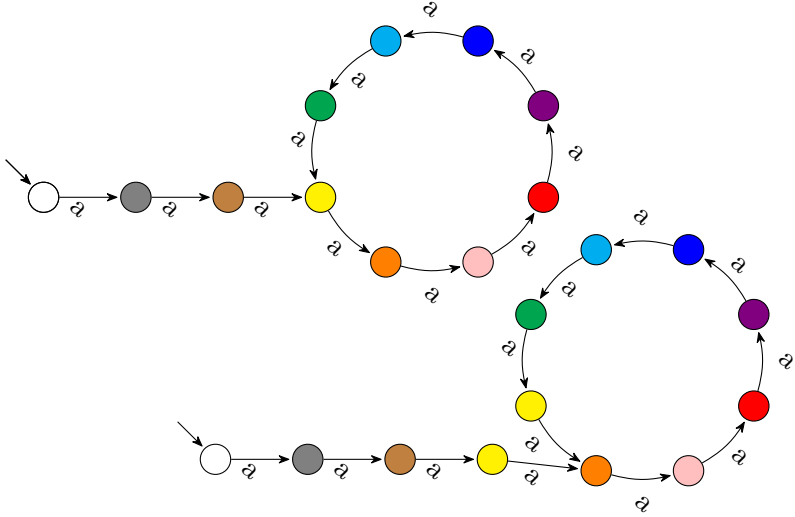

\begin{rem}
If in some coordinate the automaton were complete then the result of the intersection would be complete in that specific coordinate. To make one coordinate complete  
it is  enough to add a single sink state.
\end{rem}

\begin{ej*}[Continuation of \ejref{Sc}]
 Construction of the automaton for \linebreak
 $\big((2,0)\cup(0,1)\big)^\oplus \cup(1,1)+\big((2,0)\cup(0,1)\big)^\oplus$, see \autoref{fig:automata2001112001final}.
 \end{ej*}
 
\begin{figure}[t]
\vspace*{-2cm}
\begin{subfigure}{0.5\textwidth}
 \centering
\begin{tikzpicture}[automaton, auto, node distance = 2.5cm]

\node[state, initial, accepting] (00) []{};
\node[state] (01) [right of=00]{};
 \path[->]
 (00) edge [bend left] node [above] { $(1,0)$ } (01)
 (01) edge [bend left] node [below] { $(1,0)$ } (00)

 (00) edge [loop above] node { $(0,1)$ } ()
 (01) edge [loop above] node { $(0,1)$ } ();
\end{tikzpicture}
\caption{Automaton for $\big((2,0)\cup(0,1)\big)^\oplus$.} 
 \label{fig:automata2001}
\end{subfigure}
\begin{subfigure}{.5\textwidth}
\centering
\begin{tikzpicture}[automaton, auto, node distance = 2.5cm]

\node[state, initial] (00) []{};
\node[state] (01) [right of=00]{};
\node[state] (02) [right of=01]{};
\node[state] (10) [below of=00]{};
\node[state, accepting right, accepting] (11) [right of=10]{};
\node[state] (12) [right of=11]{};
 \path[->]
 (00) edge node { $(1,0)$ } (01)
 (01) edge [bend left] node { $(1,0)$ } (02)
 (02) edge [bend left] node { $(1,0)$ } (01)
 (10) edge node { $(1,0)$ } (11)
 (11) edge [bend left] node { $(1,0)$ } (12)
 (12) edge [bend left] node { $(1,0)$ } (11)

 (00) edge node [left] { $(0,1)$ } (10)
 (01) edge node [left] { $(0,1)$ } (11)
 (02) edge node { $(0,1)$ } (12)
 (10) edge [loop below] node { $(0,1)$ } ()
 (11) edge [loop below] node { $(0,1)$ } ()
 (12) edge [loop below] node { $(0,1)$ } ();
\end{tikzpicture}

 \caption{Automaton for $(1,1)+\big((2,0)\cup(0,1)\big)^\oplus$.} 
 \label{fig:automata112001}
 \end{subfigure}

\begin{subfigure}{.5\textwidth}
 \centering
\begin{tikzpicture}[automaton, auto, node distance = 2.5cm]

\node[state, initial] (00) []{};
\node[state, accepting above, accepting] (01) [right of=00]{};
\node[state] (02) [right of=01]{};
\node[state] (10) [below of=00]{};
\node[state] (11) [right of=10]{};
\node[state] (12) [right of=11]{};
 \path[->]
 (00) edge node { $(1,0)$ } (01)
 (01) edge [bend left] node { $(1,0)$ } (02)
 (02) edge [bend left] node { $(1,0)$ } (01)
 (10) edge node { $(1,0)$ } (11)
 (11) edge [bend left] node { $(1,0)$ } (12)
 (12) edge [bend left] node { $(1,0)$ } (11)

 (00) edge node [left] { $(0,1)$ } (10)
 (01) edge node [left] { $(0,1)$ } (11)
 (02) edge node { $(0,1)$ } (12)
 (10) edge [loop below] node { $(0,1)$ } ()
 (11) edge [loop below] node { $(0,1)$ } ()
 (12) edge [loop below] node { $(0,1)$ } ();
\end{tikzpicture}
 \caption{Automaton for \\$\Big(\big((2,0)\cup(0,1)\big)^\oplus\Big)^c\cap\Big((1,1)+\big((2,0)\cup( 0.1)\big)^\oplus\Big)^c$}\label{fig:automatacomplemento2001112002}
\end{subfigure}
\begin{subfigure}{.5\textwidth}
\begin{tikzpicture}[automaton, auto, node distance = 2.5cm]

\node[state, initial, accepting above, accepting] (00) []{};
\node[state] (01) [right of=00]{};
\node[state, accepting right, accepting] (02) [right of=01]{};
\node[state, accepting left, accepting] (10) [below of=00]{};
\node[state, accepting right, accepting] (11) [right of=10]{};
\node[state, accepting right, accepting] (12) [right of=11]{};
 \path[->]
 (00) edge node { $(1,0)$ } (01)
 (01) edge [bend left] node { $(1,0)$ } (02)
 (02) edge [bend left] node { $(1,0)$ } (01)
 (10) edge node { $(1,0)$ } (11)
 (11) edge [bend left] node { $(1,0)$ } (12)
 (12) edge [bend left] node { $(1,0)$ } (11)

 (00) edge node [left] { $(0,1)$ } (10)
 (01) edge node [left] { $(0,1)$ } (11)
 (02) edge node { $(0,1)$ } (12)
 (10) edge [loop below] node { $(0,1)$ } ()
 (11) edge [loop below] node { $(0,1)$ } ()
 (12) edge [loop below] node { $(0,1)$ } ();
\end{tikzpicture}
 \caption{Final automaton, complement of \autoref{fig:automatacomplemento2001112002}\\\mbox{ }}
 \label{fig:automata2001112001}
 \end{subfigure}
 
 \caption{Construction of the automaton for 
 $\big((2,0)\cup(0,1)\big)^\oplus \cup \Big((1,1)+\big((2,0)\cup(0,1)\big)^\oplus  \Big)$.} \label{fig:automata2001112001final}
\end{figure}
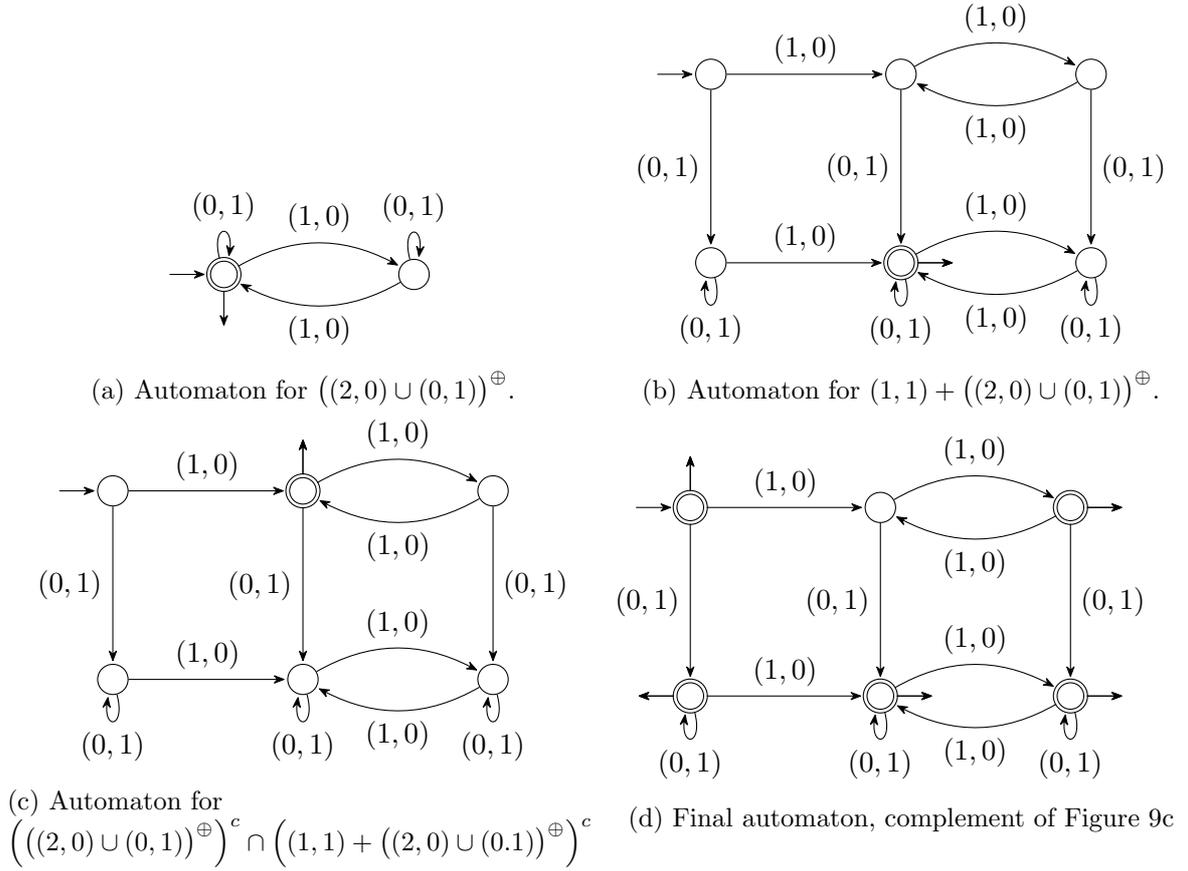

\section{From the polynomials to the resimple expression}

We provide an effective method for obtaining a resimple expression that denotes the given recognizable set, starting from the polynomials of its characteristic series.

Recall that Gohon's algorithm converts the semi-simple expression into a characteristic series that expresses as a fraction of polynomials $P'$ and $Q'$. Then it reduces this fraction by simplifying all the factors of more than one variable in $Q'$, which can be done only if the denoted set is recognizable.  
The obtained denominator $Q$  is a product of polynomials of the form $(1-x_j^{p_j})$ with  $p_j \in \mathbb{N}_{>0}$ , or $1$.
Also, it can be guaranteed that for every coordinate $j$ of $\{1,\ldots,k\}$ there is at most one factor of the form $(1-x_j^{p_j})$.

\begin{lem}[{{Gohon~\cite[Lemma 4.3]{GOHON}}}]
\lemlabel{consistent}
Let $S$ be a rational set of ${\mathbb{N}^k}$. There exists a semi-simple expression denoting $S$ and for each coordinate $j$ the following two conditions hold:
\begin{itemize}
\item There is at most one $j$-primary element in each basis of the expression.
\item If two bases in the expression contain a $j$-primary element, then it is the same in both.
\end{itemize}
\end{lem}
Furthermore, Gohon's proof gives an effective procedure for computing such a semi-simple expression. 
In the resulting expression, the non null coordinate of a $j$-primary element  is the least common multiple of all the $j$-th coordinates of the $j$-primary elements in the original expression.

\begin{lem}
\lemlabel{polyrec}
 Let $S $ a recognizable set of  $\mathbb{N}^k)$. Then,  we can compute from a semi-simple expression of $S$ two polynomials $P$ and $Q$ $\in \mathbb{Z}[x_1,x_2,...,x_k]$ such that:
 \begin{itemize}
 \item[]
 $\underline{S} = P/Q$,
 \item[]
 $Q=1$ or 
 $\Big(
 \exists J \subseteq \{1, 2, \ldots, k\},
 Q = \prod_{j \in J} (1 - x_j^{p_j}), 
 p_j \in \mathbb{N} \setminus \{0\}
 \Big)$.
 \end{itemize}
\end{lem}
\begin{proof}
Let $E$ be a semi-simple expression of $S$.
By \lemref{consistent} we can compute an equivalent, $E'$, such that for every $j \in \{1, 2, \ldots, k\}$ there exists at most one $j$-primary generator amongst all bases in $E'$. From $E'$ we compute $\underline{S}$ using  \propref{serie}. We obtain two polynomials $P'$ and $Q'$ such that $\underline{S}=P'/Q'$ and $Q'$ is a product of polynomials of the form $(1-x_1^{b_1}\cdots x_k^{b_k})$.
For each $j\in\{1,\ldots,k\}$ there exists at most one polynomial of the form $(1-x_j^{p_j})$ and by \propref{gohon} we can deduce that $P'$ is divisible by the polynomials of more than one variable that constitute $Q'$. After this simplification we obtain $P$ and $Q$.
\end{proof}

Now we show that for a recognizable set $S$ of ${\mathbb N}^k$ we can give a consistent resimple expression for $S$

\begin{thm}
\thmlabel{resimple}
Let $S$ be a recognizable set of $\mathbb{N}^k$. Then there exists a consistent resimple expression that denotes $S$ and it can be effectively computed from any semi-simple rational expression that denotes $S$.
\end{thm}
\begin{proof}
From a semi-simple expression for $S$ we obtain two polynomials $P$ and $Q$ as in \lemref{polyrec}.

If $Q=1$ then $S$ is finite and we can calculate a resimple expression that denotes it by directly applying  \propref{serie}.
In this case, the characteristic series for $S$ coincides with the polynomial $P$ and thus can be written as
$$P = \sum_{1 \leq h \leq m} \underline{d_h}.$$
with each $\underline{d_h}$ being a monomial of the form $x_1^{\delta_1} x_2^{\delta_2} \cdots x_k^{\delta_k}$.
By applying \propref{serie} inversely, we can transform each $\underline{d_h}$ into a matching expression $d_h\in {\mathbb N}^k$, and we  make the union of them. The resimple expression in this case is:
$$
\bigcup_{1 \leq h \leq m} d_h.
$$

Now consider the case  $Q\neq 1$. So by \lemref{polyrec} there exists a set $J \subseteq \{1, 2, \ldots, k\}$ such that $Q = \prod_{j \in J} (1 - x_j^{p_j})$, and 
without loss of generality, we can assume that $J = \{1, 2, \ldots, k'\}$ with $1 \leq k' \leq k$. In this case $P$ can be written as
$$P = \sum_{1 \leq h \leq m} \mu_h \underline{d_h}$$
where 
{$ \mu_h\in {\mathbb Z}- \{0\}$} and $d_h \in \mathbb{N}^k$, for each $h=1, ..., m$.  
Note that unlike the finite case there may be negative terms here.
For each coordinate $j \in J$, we define the primary element $b_j = p_j\ \mathbf{e}_j$ where  $p_j$ coincides with the corresponding  exponent in $Q$. Furthermore, for each $h=1,...,m$, , we denote $S_h = d_h + b_1^\oplus + b_2^\oplus + \cdots + b_{k'}^\oplus$. From these definitions we have
$$\underline{S_h} = \underline{d_h}/Q.$$
Then,
\[
\underline{S} = \sum_{1 \leq h \leq m} \mu_h \underline{S_h} .
\]
Consider the equivalence relation $\sim$ over $\mathbb{N}^k$ defined as
\[
s \sim s' \iff 
\left( \text{ for every } h=1, .., m,\quad
\ s \in S_h \iff  s' \in S_h
\right).
\]
Note that each $S_h$ is an atomic resimple expression, thus each equivalence class is defined by a resimple expression and, 
 by definition of the series, \defref{charseries},
we also have
\[
s \sim s' \iff \left( \text{ for every } h=1,.., m, \quad \underline{S_h} [s] = \underline{S_h}[s']\right).
\]
Then,
$s \sim s' $ exactly when 
 $\underline{S} [s] = \underline{S} [s']$. 
Notice that the number of equivalence classes is $2^m$ because for every $h=1,.., m,$ we have $\underline{S}_h[s] \in \{0,1\}$.
To refer to an   equivalence class of $\sim$ we write
\[
\s=\{ s'\in {\mathbb N}^k: s\sim s'\}.
\]
Each equivalence class $\s$ is denoted by the following  resimple expression
\begin{gather*}
C_{\s}=T_1 \cap T_2 \cap \cdots \cap T_m,\quad
\text{ where  for }h=1,..,m,\quad
T_h =
\begin{cases}
S_h, & \text{if } \underline{S}_h[s] = 1\\
(S_h)^c, & \text{if } \underline{S}_h[s] = 0.\\
\end{cases}
\end{gather*}
We are only interested in  the   equivalence classes $\s$  such that $s\in S$. 
Hence, we must take the union of the resimple expressions $C_{\s}$ just for those $\s$ such that $\underline{S}[s]=1$.
Thus, the resimple expression for $S$~is 
\[
\bigcup_{\text{ good }\s} C_{\s}
\]
where 
$\s$ is good if 
$\underline{S}[s]=1 $;
equivalently, $\s$ is good  if 
$\sum_{h=1}^m \mu_h \underline{S_h}[s] = 1$.
Furthermore, all 
the atomic resimple  expressions $S_h$  have the same basis $B = \{b_1, b_2, \ldots, b_{k'}\}$, so the given resimple expression is consistent.
\end{proof}

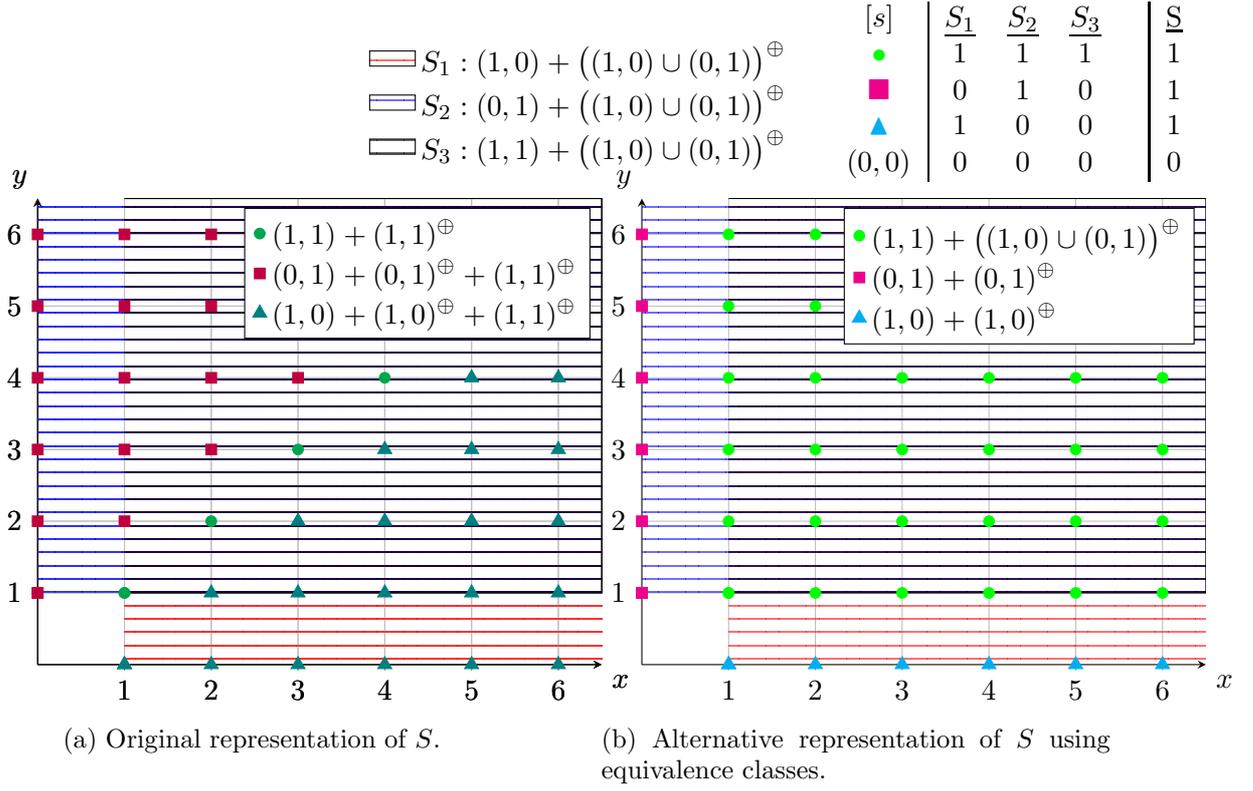
\begin{figure}[t!]
\vspace*{-1.5cm}
\begin{subfigure}{.43\textwidth}
\usetikzlibrary{patterns}
\begin{tikzpicture}
\begin{axis}[
 axis x line=center,
 axis y line=center,
 xmin=0,
 ymin=0,
 xmax=6.5,
 ymax=6.5,
 xlabel={$x$},
 xlabel style={below right},
 ylabel={$y$},
 ylabel style={above left},
 legend columns=1,
 legend cell align=left,
 legend style={ at= {(1.35,1.36)}, draw=none},
 grid=major,
 transpose legend
 ]

\addplot [draw=none,pattern = {Lines[angle=90, distance=5pt ]}, pattern color=red,
 area legend
] coordinates {
 (1, 0)
 (6.5, 0)
 (6.5,6.5)
 (1, 6.5)
};
\addlegendentry{$S_1:(1,0)+\big((1,0)\cup(0,1)\big)^\oplus $}

\addplot [draw=none,pattern = {Lines[ distance=5pt ]}, pattern color=blue,
 area legend,
] coordinates {
 (0,1)
 (6.5, 1)
 (6.5,6.5)
 (0, 6.5)
};
\addlegendentry{$S_2: (0,1)+\big((1,0)\cup(0,1)\big)^\oplus$};

\addplot [pattern = {Lines[angle=45, distance=5pt]},
 area legend,
] coordinates {
 (1,1)
 (6.5, 1)
 (6.5,6.5)
 (1, 6.5)
};
\addlegendentry{$S_3:(1,1)+\big((1,0)\cup(0,1)\big)^\oplus$};
\end{axis}
\begin{axis}[
 axis x line=center,
 axis y line=center,
 xmin=0,
 ymin=0,
 xmax=6.5,
 ymax=6.5,
 xlabel={$x$},
 xlabel style={below right},
 ylabel={$y$},
 ylabel style={above left},
 legend columns=1,
 legend cell align=left,
 grid=major,
 transpose legend
 ]

\addplot [only marks, Green] coordinates {(1,1)};
\addlegendentry{$(1,1)+(1,1)^\oplus$};
\addplot [mark=square* , only marks, purple] coordinates {(0,1)};
\addlegendentry{$(0,1)+(0,1)^\oplus+(1,1)^\oplus$};
\addplot [mark=triangle*, mark size=3pt, only marks, teal] coordinates {(1,0)};
\addlegendentry{$(1,0)+(1,0)^\oplus+(1,1)^\oplus$};
\addplot [draw=none,pattern = {Lines[angle=90, distance=5pt ]}, pattern color=red,
 area legend
] coordinates {
 (1, 0)
 (6.5, 0)
 (6.5,6.5)
 (1, 6.5)
};
\addplot [draw=none,pattern = {Lines[ distance=5pt ]}, pattern color=blue,
 area legend,
] coordinates {
 (0,1)
 (6.5, 1)
 (6.5,6.5)
 (0, 6.5)
};
\addplot [pattern = {Lines[angle=45, distance=5pt]},
 area legend,
] coordinates {
 (1,1)
 (6.5, 1)
 (6.5,6.5)
 (1, 6.5)
};
\foreach \i in {1,...,6} {
 \foreach \j in {1,...,\i}
 \addplot [mark=square*, only marks, purple] coordinates {(\j-1,\i)};
 }

\foreach \i in {1,...,6} {
 \foreach \j in {1,...,\i}
 \addplot [mark=triangle*, mark size=3pt, only marks, teal] coordinates {(\i,\j-1)};
 }

\foreach \i in {1,...,6} {
 \addplot [only marks, Green] coordinates {(\i,\i)};
 }

\end{axis}
\end{tikzpicture}
 \caption{Original representation of $S$.\\\mbox{ }}
 \label{fig:grafosin0}
\end{subfigure} \hspace*{1cm }
\begin{subfigure}{.43\textwidth}
\usetikzlibrary{patterns}
\begin{tikzpicture}
\begin{axis}[
 axis x line=center,
 axis y line=center,
 xmin=0,
 ymin=0,
 xmax=6.5,
 ymax=6.5,
 xlabel={$x$},
 xlabel style={below right},
 ylabel={$y$},
 ylabel style={above left},
 legend columns=1,
 legend cell align=left,
 grid=major,
 transpose legend
 ]
\addplot [only marks, green] coordinates {(1,1)};
\addlegendentry{$(1,1)+\big((1,0)\cup(0,1)\big)^\oplus$};
\addplot [mark=square* , only marks, magenta] coordinates {(0,1)};\addlegendentry{$(0,1)+(0,1)^\oplus$};
\addplot [mark=triangle*, mark size=3pt, only marks, cyan] coordinates {(1,0)};
\addlegendentry{$(1,0)+{(1,0)}^\oplus$};

\addplot [draw=none,pattern = {Lines[angle=90, distance=5pt ]}, pattern color=red,
 area legend
] coordinates {
 (1, 0)
 (6.5, 0)
 (6.5,6.5)
 (1, 6.5)
};

\addplot [draw=none,pattern = {Lines[ distance=5pt ]}, pattern color=blue,
 area legend,
] coordinates {
 (0,1)
 (6.5, 1)
 (6.5,6.5)
 (0, 6.5)
};

\addplot [pattern = {Lines[angle=45, distance=5pt]},
 area legend,
] coordinates {
 (1,1)
 (6.5, 1)
 (6.5,6.5)
 (1, 6.5)
};

 \foreach \i in {1,...,6} {
 \addplot [mark=square*, only marks, magenta] coordinates {(0,\i)};
 }

\foreach \i in {1,...,6} {
 \addplot [mark=triangle*, mark size=3pt, only marks, cyan] coordinates {(\i,0)};
 }

\foreach \i in {1,...,6} {
\foreach \j in {1,...,6}
 \addplot [only marks, green] coordinates {(\i,\j)};
 }
\end{axis}

\node[draw=none] at (4.9,7.6) {
 \begin{tabular}{c|cccc|c}
 $[s]$ & $\underline{S_1}$ & $\underline{S_2}$ & $\underline{S_3}$ && \underline{S}\\
 \color{green}{$\bullet$} &1&1&1 &&1\\
 \color{magenta}{$\blacksquare$} &0&1&0&&1\\
 \color{cyan}{$\blacktriangle$} &1&0&0&&1\\
 $(0,0)$ &0&0&0 &&0\\
 \end{tabular}};
\end{tikzpicture}
 \caption{Alternative representation of $S$ using equivalence classes.}
 \label{fig:grafosin0rec}
\end{subfigure}

 \caption{Two representations of $S=(S_1 \cap S_2 \cap S_3) \cup
(S_1 \cap {S_2}^c \cap {S_3}^c) \cup
({S_1}^c \cap S_2 \cap {S_3}^c)$.}
 \end{figure}

\begin{figure}[h!]
\vspace*{-0.6cm}
\begin{subfigure}{.3\textwidth}
\begin{tikzpicture}[automaton, auto, node distance = 2.5cm]

\node[state, initial] (00) []{};
\node[state, accepting] (10) [right of=00]{};
 \path[->]
 (00) edge node { $(1,0)$ } (10)
 (00) edge [loop above] node { $(0,1)$ } ()
 (10) edge [loop above] node { $(0,1),(1,0)$ } ();
\end{tikzpicture}
 \caption{$\mathcal{A}_1$ for  $S_1$}
 \label{fig:automataS1}
\end{subfigure}
\hspace*{0.25cm}
\begin{subfigure}{.3\textwidth}
 \centering
\begin{tikzpicture}[automaton, auto, node distance = 2.5cm]

\node[state, initial] (00) []{};
\node[state, accepting] (01) [below of=00]{};
 \path[->]
 (00) edge [loop right] node { $(1,0)$ } ()
 (00) edge node [right] { $(0,1)$ } (01)
 (01) edge [loop right] node { $(0,1),(1,0)$ } ();
\end{tikzpicture}
 \caption{$\mathcal{A}_2$ for $S_2$}
 \label{fig:automataS2}
\end{subfigure}
\begin{subfigure}{.3\textwidth}
 \centering
\begin{tikzpicture}[automaton, auto, node distance = 2.5cm]

\node[state, initial] (00) []{};
\node[state] (10) [right of=00]{};
\node[state] (01) [below of=00]{};
\node[state, accepting] (11) [right of=01]{};
 \path[->]
 (00) edge node [below] { $(1,0)$ } (10)
 (01) edge node { $(1,0)$ } (11)
 (10) edge [loop right] node { $(1,0)$ } ()

 (00) edge node [left] { $(0,1)$ } (01)
 (10) edge node { $(0,1)$ } (11)
 (01) edge [loop below] node { $(0,1)$ } ()
 (11) edge [loop right] node { $(0,1),(1,0)$ } ();
\end{tikzpicture}
 \caption{$\mathcal{A}_3$ for  $S_3$ }
 \label{fig:automataS1S2S3}
\end{subfigure}

\begin{subfigure}{.3\textwidth}
\begin{tikzpicture}[automaton, auto, node distance = 2.5cm]

\node[state, initial] (00) []{};
\node[state, accepting above, accepting] (10) [right of=00]{};
\node[state] (01) [below of=00]{};
\node[state] (11) [right of=01]{};
 \path[->]
 (00) edge node { $(1,0)$ } (10)
 (01) edge node { $(1,0)$ } (11)
 (10) edge [loop right] node { $(1,0)$ } ()

 (00) edge node [left] { $(0,1)$ } (01)
 (10) edge node { $(0,1)$ } (11)
 (01) edge [loop below] node { $(0,1)$ } ()
 (11) edge [loop below] node { $(0,1),(1,0)$ } ();
\end{tikzpicture}
 \caption{Automaton for  $S_d$.}
 \label{fig:automataS1S2cS3c}
\end{subfigure}\hspace*{0.25cm}
\begin{subfigure}{.3\textwidth}
 \centering
\begin{tikzpicture}[automaton, auto, node distance = 2.5cm]

\node[state, initial] (00) []{};
\node[state] (10) [right of=00]{};
\node[state, accepting left, accepting] (01) [below of=00]{};
\node[state] (11) [right of=01]{};
 \path[->]
 (00) edge node { $(1,0)$ } (10)
 (01) edge node { $(1,0)$ } (11)
 (10) edge [loop right] node { $(1,0)$ } ()

 (00) edge node [left] { $(0,1)$ } (01)
 (10) edge node { $(0,1)$ } (11)
 (01) edge [loop below] node { $(0,1)$ } ()
 (11) edge [loop below] node { $(0,1),(1,0)$ } ();
\end{tikzpicture}
 \caption{Automaton for $S_e$.}
 \label{fig:automataS1cS2S3c}
\end{subfigure}\hspace*{0.25cm}
\begin{subfigure}{.3\textwidth}
 \centering
\begin{tikzpicture}[automaton, auto, node distance = 2.5cm]

\node[state, initial] (00) []{};
\node[state, accepting above, accepting] (10) [right of=00]{};
\node[state, accepting left, accepting] (01) [below of=00]{};
\node[state, accepting] (11) [right of=01]{};
 \path[->]
 (00) edge node { $(1,0)$ } (10)
 (01) edge node { $(1,0)$ } (11)
 (10) edge [loop right] node { $(1,0)$ } ()

 (00) edge node [left] { $(0,1)$ } (01)
 (10) edge node { $(0,1)$ } (11)
 (01) edge [loop below] node { $(0,1)$ } ()
 (11) edge [loop below] node { $(0,1),(1,0)$ } ();
\end{tikzpicture}
\caption{Automaton for $S$}
\label{fig:automataS}
\end{subfigure}
\caption{
Construction of the automaton for $S= {(S_1 \cap S_2 \cap S_3)}^c\cap{(S_1 \cap {S_2}^c \cap {S_3}^c)}^c\cap{({S_1}^c \cap S_2 \cap {S_3}^c)}^c$, where 
$S_1 = (1,0)+\big((1,0)\cup(0,1)\big)^\oplus$,
$S_2 = (0,1)+\big(1,0)\cup(0,1)\big)^\oplus$,\linebreak
$S_3 = (1,1)+\big((1,0)\cup(0,1)\big)^\oplus= S_1 \cap {S_2} \cap {S_3}$, $S_d= S_1 \cap {S_2}^c \cap {S_3}^c$, $S_e={S_1}^c \cap S_2 \cap {S_3}^c$.
}
\end{figure}
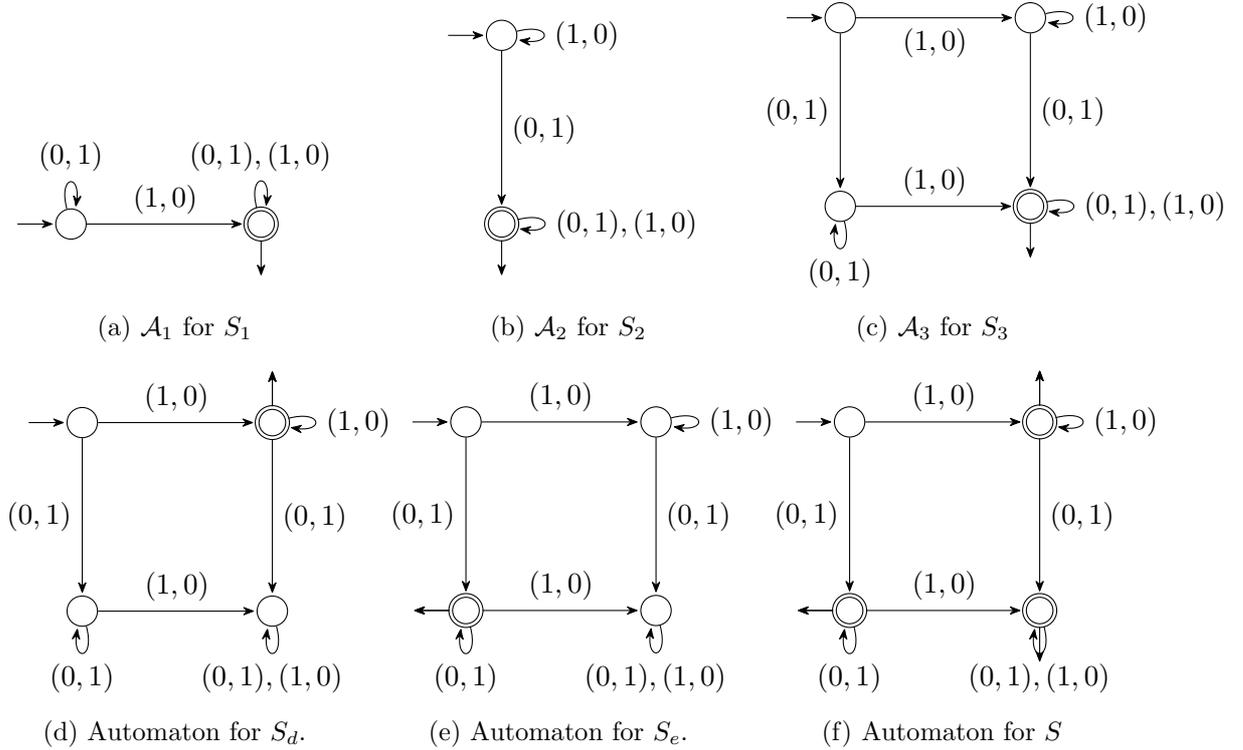
\begin{rem}
The  automaton obtained from the proof of \thmref{automata} is not always minimal. For example, the one in \autoref{fig:automataS} is not minimal. The minimal automaton is obtained by applying Moore's algorithm and it is shown  in Figure~\ref{fig:automataminS}. 
\end{rem}

\begin{figure}[ht]
 \centering
\begin{tikzpicture}[automaton, auto, node distance = 3.5cm]

\node[state, initial] (00) []{};
\node[state, accepting] (10) [right of=00]{};
 \path[->]
 (00) edge node { $(0,1),(1,0)$ } (10)
 (10) edge [loop right] node { $(0,1),(1,0)$ } ();
\end{tikzpicture}
 \caption{Minimum automaton for $S$.}
 \label{fig:automataminS}
\end{figure}
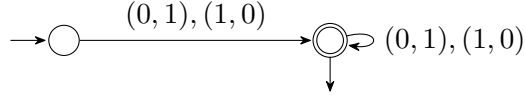

\begin{ej}
\ejlabel{S}
Let us consider $S$ the recognizable set depicted in Figure~\ref{fig:grafosin0}, denoted by the semi-simple expression
$S = (1,1)+(1,1)^\oplus\ \cup
(1,0)+\big((1,0)\cup(1,1)\big)^\oplus \ \cup
(0,1)+\big((0,1)\cup(1,1)\big)^\oplus.$
Its characteristic series is:
{\small 
\begin{align*}
\underline{S} &= 
\frac{xy}{(1-xy)} + \frac{x}{(1-xy)(1-x)} + \frac{y}{(1-xy)(1-y)}\\
&= \frac{xy(1-x)(1-y) + x(1-y) + y(1-x)}{(1-xy)(1-x)(1-y)}
\\
&= \frac{xy - x^2y - xy^2 + x^2y^2 + x - xy + y - xy}{(1-xy)(1-x)(1-y)}
\\
&= \frac{(1-xy)(x+y-xy)}{(1-xy)(1-x)(1-y)}
\\
&= \frac{x+y-xy}{(1-x)(1-y)}.
 \end{align*}
Then, we  separate each term of the numerator and we obtain $\underline{S_1},\underline{S_2}$ and $\underline{S_3}$. We omit the minus for the last term in order to give the resimple expressions, but it is considered afterwards for the equivalence classes.

\begin{align*}
\underline{S_1} &= \frac{x}{(1-x)(1-y)}
&\underline{S_2} &= \frac{y}{(1-x)(1-y)}
&\underline{S_3} &= \frac{xy}{(1-x)(1-y)}
\\
S_1 &= (1,0)+\big((1,0)\cup(0,1)\big)^\oplus \quad
&S_2 &= (0,1)+\big((1,0)\cup(0,1)\big)^\oplus \quad
&S_3 &= (1,1)+\big((1,0)\cup(0,1)\big)^\oplus.
\end{align*}
Finally, we find the corresponding resimple expression $C$,  illustrated in Figure~\ref{fig:grafosin0rec}.
This expression comes from considering three equivalence classes: elements in $S_1$, $S_2$ and $S_3$; in $S_1$ but not in the other two; and in $S_3$ but not in the other two. The other equivalence classes are not considered because its  elements are not in $S$, or because  the class is empty.
Our algorithm stops here.

For the sake of this example 
we check that the characteristic series defined from $C$  and the characteristic series  of $S$ coincide.
\begin{align*}
C&= &(S_1 \cap S_2 \cap S_3)& &\cup&
&(S_1 \cap {S_2}^c \cap {S_3}^c)& &\cup&
&({S_1}^c \cap S_2 \cap {S_3}^c)\\
&= &(S_3)& &\cup&
&(S_1 \cap {S_2}^c)& &\cup&
&(S_1^c \cap S_2)&\\
&= &(1,1)+\big((1,0)\cup(0,1)\big)^\oplus& \quad &\cup& \quad
&(1,0)+{(1,0)}^\oplus& \quad &\cup& \quad
&(0,1)+{(0,1)}^\oplus&
\\\bigskip
\underline{S}&= &\frac{xy}{(1-x)(1-y)}& &+&
&\frac{x}{(1-x)}& &+&
&\frac{y}{(1-y)}&\\
&= &\frac{xy}{(1-x)(1-y)}& &+&
&\frac{x(1-y)}{(1-x)(1-y)}& &+&
&\frac{y(1-x)}{(1-x)(1-y)}&\\
 & =&\frac{xy +
x(1-y) +
y(1-x)}{(1-x)(1-y)}.
\end{align*}
}\samepage
After obtaining the resimple expression $C$ we are able to construct the automaton.
\end{ej}

\section{Algorithm}
\label{algo}

We prove \thmref{main} by giving the following algorithm. Let $A$ be an alphabet of size $k$.
\bigskip

\noindent
{\it Input}: $E$ a regular expression over $A^*$.\medskip
\\
\noindent
{\it Output}: If ${\EuScript C}(E)$ is regular then  the output is  a complete finite state automaton $\mathcal{A}$ over $A^*$ that accepts
${\EuScript C}(E)$, the commutative closure of $E$.
Otherwise, the algorithm stops and  warns that ${\EuScript C}(E)$ is not regular.

\begin{enumerate}

\item Obtain $E'$ a semi-simple consistent expression over $\mathbb{N}^k$  for $\varphi(E)$ applying  \propref{semisimple} (to make it semi-simple) and \lemref{consistent} (to make it consistent).
\item {Obtain $P'$ and $Q'$ two polynomials  in $ \mathbb{Z}[x_1,...,x_k]$ such that $\underline{E'} = \frac{P'}{Q'}$}, applying \propref{serie} to $E'$.

\item Obtain $P$ and $Q$ in $ \mathbb{Z}[x_1,...,x_k]$ such that $\frac{P}{Q}$ irreducible by simplifing all possible factors of $\frac{P'}{Q'}$.
\item If $Q$ has any factor of more than one variable then $E$ it is not recognizable. The algorithm stops and  warns that ${\EuScript C}(E)$ is not regular.
\item Obtain $C$ a consistent resimple expression over $\mathbb{N}^k$  from the polynomials $P$ and $Q$ .
 applying \thmref{resimple}.
 \item Obtain the automaton $\mathcal{A}$ applying \thmref{automata} to the expression $C$ from the previous step.
\end{enumerate}

\section{Complexity}

We use the  asymptotic big O notation asserting that for  functions $f,g:{\mathbb N}\to{\mathbb R}$,
$f(n)$ is $\bigo(g(n))$ if there exists $C$ such that for all sufficiently large $n$ $|g(n)|\leq |Cf(n)|$.

We assume an alphabet of $k$ letters and use the following notation to refer to various size functions.
For any set $B$, the number of elements of $B$ is $|B|$. 
For any $\sigma = (b_1, \ldots , b_k) \in \mathbb{N}^k$,  $||\sigma||= \max_{1\leq j \leq k} b_j$
; similarly, for any $B \subseteq \mathbb{N}^k$ we denote $||B|| = \max_{\sigma\in B} ||\sigma||$.
Finally, for a semi-linear expression $E = \bigcup\limits_{i \in I}\gamma_i + B_i^\oplus$  we write $||E|| = \max(\max_{i\in I} ||\gamma_i|| , \max_{i \in I} ||B_i||, 2)$.

\subsection{State complexity}
By the state complexity of a regular language we mean the minimum number of states of a complete  automaton that accepts that language~\cite{YU}. It is a natural measure for operations on regular languages and, in turn, it gives us a lower bound for the temporal and spatial complexity of operations on automata.
The study of state complexity dates back at least to~\cite{MASLOV}.

\begin{prop}[\cite{HASSE}]
\proplabel{desambiguar}
    Every semi-linear set denoted by a semi-linear expression $E = \bigcup\limits_{i \in I}\gamma_i + B_i^\oplus$ has an equivalent semi-simple expression $E' = \bigcup\limits_{i \in I'}\gamma'_i + {B'_i}^\oplus$ where
 $$   ||\gamma'_i|| \leq ||E||^{|I|\cdot\bigo(k^6)}, \quad   ||B'_i|| \leq ||E||^{|I|\cdot\bigo(k^4)}, \quad     |I'| \leq ||E||^{\bigo(k^5)}.
    $$
\end{prop}
Although our algorithm starts from a regular expression in $A^*$, we  give the bound on the semi-simple expression of the set of $\mathbb{N}^k$.
Let a semi-simple expression be
$$E = \bigcup\limits_{i \in I}\gamma_i + B_i^\oplus.$$

\begin{define}[Value $p_j$]
\deflabel{pj}
Given a semi-simple expression  $E=\bigcup\limits_{i \in I}\gamma_i + B_i^\oplus$,  for each $j=1, \ldots k$,
we define $p_j = lcm (m_1, \ldots, m_{|I|})$
    where $$
m_i = \begin{cases}
m & \text{ if } m \mathbf{e}_j \in B_i\\
1 & \text{ otherwise}.
\end{cases}$$
\end{define}

\begin{prop}
Let   $E=\bigcup\limits_{i \in I}\gamma_i + B_i^\oplus$.
For each $j \in \{1,\ldots,k\}$, $p_j = \bigo \Big(e^{\sqrt{n.log(n)}}\Big)$, where $n$ is the size of the semi-simple expression $E$, that is $n=|I| \; ||E||$. \end{prop}
\begin{proof}
Let   $p_j = lcm (m_1, \ldots, m_{|I|})$, where
$m_i$ are in \defref{pj}. So,  $m_1+\ldots +m_{|I|}~\leq~n$. 
Let 
$F(n)=\max\{lcm(o_1, \ldots, o_l):o_1 +\cdots+o_l = n, l \in \mathbb{N}\}$.
Then,  
\[
p_j = lcm (m_1, \ldots, m_{|I|})\leq F(n).
\]
The problem of finding a good approximation for $F(n)$ is known as Landau's
problem. It is well studied and, as shown in \cite{SZALAY}, 
$F(n)= \bigo \Big(e^{\sqrt{n.log(n)}}\Big)$.

\end{proof}

\begin{prop}\label{prop:complexity-state}
The state complexity of the automaton constructed by our algorithm, taking a semi-simple expression $E$, is at most $\bigo \Big(e^{k\sqrt{n.log(n)}}\Big)$ with $n$ the size of $E$.
\end{prop}

\begin{proof}

When introducing resimple expressions and their automata, we observed that the maximum number of states per coordinate does not increase when performing boolean operations between automata derived from consistent resimple expressions (\propref{inter}). Therefore, starting from a consistent semi-simple expression
$
E' = \bigcup\limits_{i \in I'}\gamma'_i + B_i^\oplus
$
and using \propref{cotaautomata}, the number of states per coordinate can be bounded by
\[
l_j=\max_{i \in I'}\{|\gamma'_i|_j\}+p_j-1.
\]
Then, the state complexity of an automaton derived from a semi-simple consistent expression $E'$ is at most
\[
\prod_{j=1}^k l_j=\prod_{j=1}^k(\max_{i \in I}\{|\gamma'_i|_j\}+p_j-1)= \bigo\big ((\gamma_{max}+p_{max})^k\big),
\]
where
$p_{max} = \max _{j \in \{1,\ldots,k\}} p_j = \bigo \Big(e^{\sqrt{n.log(n)}} \Big)$,
$\gamma_{max}=\max_{i \in I }|\gamma_i| = \bigo(n)$.
Clearly $\gamma_{max} \leq ||E'||$, then
$\gamma_{max}$ can be bounded by $n$.

To transform the semi-simple expression $E$ into a consistent $E '$ we need to change each of the $j$-primary bases by $p_j\; \mathbf{e}_j$.
In the worst case,
for each coordinate, we must add
to some term $c$
up to $(p_j-1) \mathbf{e}_j$ .
Then
$
\max_{i \in I}\{|\gamma'_i|_j\} \leq \max_{i \in I}\{|\gamma_i|_j\}+(p_j-1).
$
This does not alter the total state complexity, \begin{align*}
\prod_{j=1}^k(\max_{i \in I}\{|\gamma_i|_j\}+2(p_j-1 ))= \bigo((\gamma_{max}+p_{max})^k)
=\bigo\bigg(\Big(n+e^{\sqrt{n.log(n)}} \Big)^k\bigg).
\end{align*}
Thus, the total state complexity is  $\bigo\bigg(\Big(e^{\sqrt{n.log(n)}}\Big)^k\bigg).$
\end{proof}

\begin{rem} If we choose to reduce the polynomials as much as possible, we  reach the smallest possible $p_j$, so in general we obtain an automaton as small as possible in each of its coordinates. However,
as we have already seen in \ejref{S},
even if the algorithm starts
from irreducible $P/Q$,
the resulting automaton is not necessarily the minimum automaton.
\end{rem}

\begin{rem}
In~\cite{HOFFMANN} Hoffmann 
 proves for the case of group languages a state complexity of $\bigo (n^k e^{k\sqrt{n.log(n)}})$, with $n$ being the number of states of the permutation automaton.
\end{rem}

\subsection{Time complexity}

We call elementary operations any arithmetic operation on natural, rational or real numbers.

\begin{prop}\label{prop:complexity-time}
Our algorithm has a time complexity of $\bigo(m^2 2^m + m^{7k})$ elementary operations in the worst case, where $m = \bigo\Big(e^{k \sqrt{n.log(n)}}\Big)$, $n$ is the size of the semi-simple expression $E$ and $k$ is the size of the alphabet.
\end{prop}

\begin{proof}
To transform the semi-simple expression $E$ into a consistent one $E'$ we need to change all the $j$-primary bases to $p_j \mathbf{e}_j$.
For that we need to consider at most $np_1\cdots\ p_k$ simple terms that we  use to build the expression $E'$ according to the \lemref{consistent}.

Then $E'$  consists of $t=\bigo\Big(|I|\prod_{j=1}^kp_j\Big)$ simple terms. Since $|I|$
is a disjoint union and we assume a fixed alphabet, we can limit $\bigo ((\gamma_{max}+p_{max})^k)$, otherwise we would have repeated terms. Also, $\bigo\big(\prod_{j=1}^kp_j \big) = \bigo\big({p_{max}}^k\big)$. Then,
\[t=\bigo\Big(((\gamma_{max}+p_{max})p_{ max})^k\Big)=\bigo(m^2).
\]
When converted to a series, each term of the consistent semi-simple expression generates a fraction. By taking a common denominator, in the worst case, we  have to multiply each term by the denominator of the remaining ones.
Note that
each denominator
has at most
$k$ factors of the form $(1-\underline{d} )$, since the basis from which they come
has at most $k$ elements.

If we distribute each of the denominators,
the polynomial
will have at most $2^k$ terms. So, when taking a common denominator, we have to multiply each numerator by the $(t-1)$ remaining denominators.
Therefore, there remain
$2^k(t-1)t$ terms in the numerator. Note that $\bigo\Big(2^kt^2\Big)=\bigo( t^2)$.

Now we need to simplify the denominator factors of more than one variable.
{We do it by reducing }$P'/Q'$. 
For this we factor $P'$ and $Q'$ with cost $\bigo(m^{7k}), see $\cite{L3}.
This upper bound is because we can bound the degree of each variable of the polynomial by $\bigo (\gamma_{max} + p_{max})$ and the coefficients, as they are from a characteristic series, are $1$ or $-1$.

In the recognizable case, after factoring and reducing, there cannot be more than $m$ terms. This is because all the  resimple expressions have the same infinite part and 
(by the same argument we used to  to bound $|I|$),
there are most  $(\gamma_{max}+p_{max})^k$  termns.
We  have at most $m$ automata to intersect. As we already mentioned, the maximum number of states of the intersection automaton is $\bigo(m)$. 
Notice that the number of subclasses defined in the construction of the resimple expression is at most $\bigo(2^m)$. Thus, we have $\bigo(m2^m)$ boolean operations, each with linear cost in $m$.
Then, the total cost of computing the final automaton from the polynomials is $\bigo(m^2 2 ^m)$.
We obtain a total worst-case time complexity
$\bigo(m^2 2^m + m^{7k}+ m^4)$ which is
$\bigo(m^2 2^m + m ^{7k})$.
\end{proof}

\section{Acknowledgements}
We thank Jacques Sakarovitch for insightful discussions at an early stage of this work.
This research was supported by a grant from the University of Buenos Aires.

\bibliographystyle{plain}
\bibliography{submittedJCSS20Abril2025}
\bigskip

\begin{tabular}{ll}
Verónica Becher&
{\tt vbecher@dc.uba.ar}
\\
Simón Lew Deveali
&{\tt sdeveali@dc.uba.ar}
\\
Ignacio Mollo Cunningham &{\tt imcgham@gmail.com}
\end{tabular}
\bigskip

\noindent
{\small Departmento de Computación, Facultad de Ciencias Exactas y Naturales,
Universidad de Buenos Aires, e 
Instituto de Ciencias de la Computación(ICC)  de la Universidad de Buenos Aires y CONICET.
\\
Pabellón 0, Ciudad Universitaria, 
(1428) Buenos Aires, Argentina
}

\end{document}